%% file: main.tex
\pdfoutput=1
\documentclass[12pt]{article}
\input{preamble}
\usepackage[authoryear,square,semicolon]{natbib}
\usepackage[affil-it]{authblk}

\title{Dynamic Latent Space Modeling of Inhomogeneous Poisson Network Processes with Applications to International Relations}
\author[1]{Jie Jian}
\author[2]{Jiguo Cao}
\author[2]{Owen G. Ward}

\affil[1]{Data Science Institute, University of Chicago}
\affil[2]{Department of Statistics and Actuarial Science, Simon Fraser University}
\date{}

\begin{document}
\maketitle

\begin{abstract}
\input{tex_files/0_abstract}
\end{abstract}

\input{tex_files/1_introduction}

\input{tex_files/2_methodology}
\input{tex_files/3_identifiability}

\input{tex_files/4_estimation}
\input{tex_files/5_simulation}
\input{tex_files/6_application}

\input{tex_files/7_discussion}

\input{tex_files/8_AIdeclare}

{\renewcommand{\baselinestretch}{1}\selectfont
\bibliographystyle{plainnat}
\bibliography{LSMrefs}
}

\clearpage
\appendix
\section{Gauss--Legendre quadrature}
\label{supp:quad}

The rule used in (22) is the standard Gauss--Legendre rule mapped to the observation window. Let $(u_q,\omega_q)_{q=1}^{Q}$ be the nodes and weights of the $Q$-point rule on the reference interval $[-1,1]$, obtained as the roots of the degree-$Q$ Legendre polynomial and the associated weights. The affine change of variable from $[-1,1]$ to $[0,T]$ gives
\[
  s_q = \tfrac{T}{2}\,(u_q + 1),
  \qquad
  w_q = \tfrac{T}{2}\,\omega_q,
  \qquad q = 1,\dots,Q,
\]
and the weights satisfy $\sum_q w_q = T$, since the rule integrates the constant function exactly. We use $Q = 40$ points for the cumulative intensities (22) and a finer grid of $Q_{\bm{\Omega}} = 200$ points for the Gram matrices $\bm{\Omega}_1, \bm{\Omega}_2$, the latter computed once before optimization begins.


\section{Proofs for Section 3}
\label{app:ident}

Throughout, $\mathbf{Q} = \mathbf{I}_n - n^{-1} \bm{1} \bm{1}^{\top}$,
$\mathbf{D}(t) = [d_{ij}(t)^2]$ and
$\mathbf{G}(t) = -\tfrac{1}{2} \mathbf{Q} \mathbf{D}(t) \mathbf{Q}$. 

\subsection{Proof of Proposition 1} 

\begin{proof}
(i) We fix $t$ and suppress it. Since the distances remain unchanged,
\begin{align*}
d_{1i}^2  = \| z_i \|^2  = \| z^{*}_i \|^2 = (d^{*}_{1i})^2,
\end{align*}
and
\begin{align*}
\| z_i \|^2 +\| z_j \|^2 - 2 z_i^{\top} z_j  =  \| z^{*}_i \|^2 +\| z^{*}_j \|^2 - 2 z_i^{*\top} z^{*}_j ,
\end{align*}
we have
$\mathbf{Z} \mathbf{Z}^{\top} = \mathbf{Z}^{*} \mathbf{Z}^{*\top}$. There is an orthogonal map $\mathbf{U}$ gives $\mathbf{Z}^{*} = \mathbf{Z} \mathbf{U}$ with $\mathbf{U} \in O(2)$. Now, on the first two nodes:
\begin{align*}
\begin{bmatrix}
    0 & 0\\ z_{21} &0
\end{bmatrix} 
\begin{bmatrix}
u_{11} & u_{12} \\ u_{21} & u_{22}
\end{bmatrix}
= 
\begin{bmatrix}
    0 & 0\\ z_{21}^\star & 0
\end{bmatrix} ,
\end{align*}
sets $u_{12}=u_{21}=0$, $u_{11}=u_{22}=\pm 1$. 

(ii) Under (A2), the matrix $\mathbf{Z}(t)$ has full column rank for every $t$, so $\mathbf{U}(t) = \{ \mathbf{Z}(t)^{\top} \mathbf{Z}(t) \}^{-1} \mathbf{Z}(t)^{\top} \mathbf{Z}^{*}(t)$ is the unique solution of $\mathbf{Z}^{*}(t) = \mathbf{Z}(t) \mathbf{U}(t)$ and is continuous on $[0,T]$. A continuous map from a connected set into the four-point set $\mathbb{V}_4$ is constant. 

(iii) By (6) 
the trajectories $\mathbf{Z}(t) \operatorname{diag}(\epsilon_1, \epsilon_2)$ are generated by $(\epsilon_1 \mathbf{C}_x, \epsilon_2 \mathbf{C}_y)$, which satisfies (13) 
whenever $\mathbf{c}$ does. All pairwise distances are unchanged, hence so is $\ell$, and $f_1$, $f_2$ and $\| \mathbf{c} \|_F$ are quadratic in each of $\mathbf{C}_x$ and $\mathbf{C}_y$ separately, hence invariant under independent sign changes. The four arrays are distinct if and only if $\mathbf{C}_x \neq \bm{0}$ and $\mathbf{C}_y \neq \bm{0}$, which hold under (A1) and (A2) respectively.
\end{proof}


\subsection{Proof of Proposition 2}

\begin{proof}
By equations (2) and (5), the profile log-likelihood takes the form
\begin{equation}
\label{eq:beta-profile}
  \ell(\mathbf{Z}, \bm{\beta})
  \;=\;
  \sum_{i \in V} m_i \beta_i
  \;-\;
  \sum_{(i,j) \in \mathcal{D}} e^{\beta_i + \beta_j} I_{ij}
  \;+\; \text{const}.
\end{equation}
Differentiating twice yields gives (16). 
The quadratic form is non-negative and vanishes only if $v_i + v_j = 0$ for all $i \neq j$. For $n \ge 3$, subtracting $v_i + v_j = 0$ from $v_i + v_k = 0$ implies $v_j = v_k$ for all $j, k \neq i$, forcing $\mathbf{v} = v \bm{1}$ and $2v = 0$. This establishes strict concavity in (ii), and (i) follows immediately because equal intensities require $v_i + v_j = 0$ for $\mathbf{v} = \bm{\beta}^{*} - \bm{\beta}$.

For (iii), strict concavity implies that a unique finite maximizer exists if and only if $\ell(\mathbf{Z}, \bm{\beta})$ is coercive, meaning $\ell(\bm{\beta} + s \mathbf{v}) \to -\infty$ as $s \to \infty$ along every direction $\mathbf{v} \neq \bm{0}$. Define the recessive cone $\mathcal{V} = \{ \mathbf{v} \neq \bm{0} : v_i + v_j \le 0 \text{ for all } (i,j) \in \mathcal{D} \}$. 

If $\mathbf{v} \notin \mathcal{V}$, at least one dyad has $v_i + v_j > 0$, so the exponential term $-e^{\beta_i + \beta_j + s(v_i + v_j)} I_{ij}$ grows without bound and drives $\ell \to -\infty$. If $\mathbf{v} \in \mathcal{V}$, the exponential terms remain bounded as $s \to \infty$, leaving the asymptotic slope determined entirely by the linear term $s \mathbf{m}^{\top} \mathbf{v}$. Consequently, coercivity holds if and only if $\mathbf{m}^{\top} \mathbf{v} < 0$ for every $\mathbf{v} \in \mathcal{V}$.

It remains to show that $\mathbf{m}^{\top} \mathbf{v} < 0$ on $\mathcal{V}$ is equivalent to $0 < m_i < m$ for all $i$.

To establish necessity, suppose $m_i = 0$ for some $i$. Taking $\mathbf{v} = -\mathbf{e}_i \in \mathcal{V}$ gives $\mathbf{m}^{\top} \mathbf{v} = 0$, violating strict negativity. Next, suppose $m_i \ge m$ for some $i$. Taking $v_i = 1$ and $v_j = -1$ for all $j \neq i$ gives $\mathbf{v} \in \mathcal{V}$ and $\mathbf{m}^{\top} \mathbf{v} = m_i - \sum_{j \neq i} m_j = 2(m_i - m) \ge 0$, again violating strict negativity.

To establish sufficiency, assume $0 < m_i < m$ holds for all $i$ and pick any $\mathbf{v} \in \mathcal{V}$. Order coordinates so that $v_1 = \max_i v_i$. If $v_1 \le 0$, then all coordinates are non-positive and at least one is strictly negative, yielding $\mathbf{m}^{\top} \mathbf{v} < 0$ since $m_i > 0$. If $v_1 > 0$, the constraint $v_1 + v_j \le 0$ forces $v_j \le -v_1$ for all $j \ge 2$, whence
\begin{equation}
  \mathbf{m}^{\top} \mathbf{v}
  \;\le\;
  m_1 v_1 - v_1 \sum_{j=2}^n m_j
  \;=\;
  v_1 \bigl\{ m_1 - (2m - m_1) \bigr\}
  \;=\;
  2 v_1 (m_1 - m)
  \;<\;
  0.
\end{equation}
Thus $\mathbf{m}^{\top} \mathbf{v} < 0$ holds for all $\mathbf{v} \in \mathcal{V}$, ensuring coercivity and completing the proof.
\end{proof}


\subsection{Proof of Proposition 3}

\begin{proof}
(i) In (18), we have equation for pairwise intensity
\begin{equation}
\label{eq:additive-shift-supplementary}
  d^{*}_{ij}(t)^2 = d_{ij}(t)^2 + a_i + a_j, \quad (i,j) \in \mathcal{D},\ t \in [0,T],
\end{equation} 
which leads to the matrix form
\begin{equation}
\label{eq:additive-shift-matrix-supplementary}
  \mathbf{D}^{*}(t) -\mathbf{D}(t) =  \bm{a}\bm{1}^\top  + \bm{1}\bm{a}^\top - 2 \operatorname{diag}(\bm{a}).
\end{equation} 
Multiplying $Q$ to both sides of~\eqref{eq:additive-shift-matrix-supplementary} eliminates the vector $\bm{a}$
\begin{equation}
\label{eq:additive-shift-matrix-Q-supplementary}
\mathbf{Q} \mathbf{D}^{*}(t) \mathbf{Q} -\mathbf{Q} \mathbf{D}(t)\mathbf{Q} =\bm{0}  + \bm{0} - 2 \mathbf{Q} \operatorname{diag}(\bm{a})\mathbf{Q},
\end{equation}
which is Equation (19) in Proposition (i). 

The conclusion that $\operatorname{rank} \mathbf{G}(t) \le 2$ for all $t \in [0, T]$ comes from the representation of $\mathbf{D}(t)$. The squared distance between node positions $z_i(t), z_j(t) \in \mathbb{R}^2$ is $$d_{ij}(t)^2 = \Vert{}z_i(t) - z_j(t)\Vert{}^2 = \Vert{}z_i(t)\Vert{}^2 + \Vert{}z_j(t)\Vert{}^2 - 2 z_i(t)^\top z_j(t)$$Let $\mathbf{b}(t) = (\Vert{}z_1(t)\Vert{}^2, \dots, \Vert{}z_n(t)\Vert{}^2)^\top$ be the $n \times 1$ vector of squared norms. Expanding across all pairs yields the matrix equation:$$\mathbf{D}(t) = \mathbf{b}(t)\mathbf{1}^\top + \mathbf{1}\mathbf{b}(t)^\top - 2 \mathbf{Z}(t)\mathbf{Z}(t)^\top.$$ Multiplying $\mathbf{D}(t)$ on the left and right by $\mathbf{Q}$ yields
\begin{equation}
\label{eq:rank-G-supplementary}
\mathbf{G}(t) := -\tfrac{1}{2} \mathbf{Q} \mathbf{D}(t) \mathbf{Q} = \mathbf{Q} \mathbf{Z}(t) \mathbf{Z}(t)^{\top} \mathbf{Q}
\end{equation}
Therefore, $\mathbf{G}(t)$ and $\mathbf{G}^{*}(t)$ both have rank of at most 2. Consequently, $\operatorname{rank}(\mathbf{G}^*(t) - \mathbf{G}(t)) \le 4$ and $|\mathcal{S}| \le 4$.

(ii) First, we show that $a_1=0$. Let $\mathcal{U}= \{ i : a_i = 0 \}$. By definition, for any $i,j\in \mathcal{U}$, ~\eqref{eq:additive-shift-supplementary} yields $d_{ij}(t)=d_{ij}^{*} (t)$ for all $t$. Therefore, the positions $\{z_i^*(t)\}_{i \in \mathcal{U}}$ and $\{z_i(t)\}_{i \in \mathcal{U}}$ are isometric. There exist an orthogonal matrix $\mathbf{R}(t) \in O(2)$ and a translation vector $\mathbf{e}(t) \in \mathbb{R}^2$ such that 
\begin{equation}
\label{eq:z_i-supplementary}
z_i^{*}(t)=\mathbf{R}(t) z_i(t) + \mathbf{e}(t) \quad i \in \mathcal{U},\ t \in [0,T].
\end{equation}
The distance between $i\in \mathcal{U}$ and the first node satisfies $$d_{i1}^*(t)^2 = d_{i1}(t)^2 + a_1 \quad \text{for all } i \in \mathcal{U},$$
and plugging~\eqref{eq:z_i-supplementary} into it yields $$\Vert{}\mathbf{R}(t) z_i(t) + \mathbf{e}(t)\Vert{}^2 = \Vert{}z_i(t)\Vert{}^2 + a_1.$$
Expanding and canceling $\Vert{}z_i(t)\Vert{}^2$ from both sides, we get $$2 \bigl(\mathbf{R}(t)^\top \mathbf{e}(t)\bigr)^\top z_i(t) = a_1 - \Vert{}\mathbf{e}(t)\Vert{}^2 \quad \text{for all } i \in \mathcal{U}.$$ If $\mathbf{R}(t)^\top \mathbf{e}(t)$ is not $\bm{0}$, then at each time $t$, all nodes in $\mathcal{U}$ lie on the same line which contradicts (A3). Thus, $\mathbf{R}(t)^\top \mathbf{e}(t) = \bm{0}$ which implies $\mathbf{e}(t) = \bm{0}$ and $a_1=0$.

Then we show that $\bm{a}=\bm{0}$. For $k \in \mathcal{S}$ and $i \in \mathcal{U}$, we have $d^{*}_{ik}(t)^2 = d_{ik}(t)^2 + a_k$. Expanding, plugging in $z_i^*(t) = \mathbf{R}(t) z_i(t)$, and arranging the equation, we have 
\begin{equation}
    \label{eq:z_i_k-supplementary}
    2 z_i(t)^\top \bigl( z^{*}_k(t) - z_k(t) \bigr) = \Vert{}z^{*}_k (t)\Vert{}^2 - \Vert{}z_k(t)\Vert{}^2 - a_k 
\end{equation}
As node $1 \in \mathcal{U}$, we have the right-hand side of~\eqref{eq:z_i_k-supplementary} as $0$. Since $\vert{}\mathcal{U}\vert{} \ge n - 4 \ge 3$, the unshifted nodes $\{z_i(t)\}_{i \in \mathcal{U}}$ contain at least 3 non-collinear points, which implies $z^{*}_k(t) - z_k(t)= \bm{0}$, and thus $a_k=0$. Given $\bm{a}=0$, by Proposition 1, under the anchoring system, we also know that the two positions  $\{z_i^*(t)\}$ and $\{z_i(t)\}$ are identifiable up to $\mathbb{V}_4 $.
\end{proof}

\subsection{Proof of Proposition 4}

\begin{proof}
Let $\mathbf{u} = (\mathbf{u}_{\bm{\beta}}, \mathbf{u}_{\mathbf{c}}) \in \mathbb{R}^{n + (2n-3)K}$ be a tangent vector in the coordinate space of $\mathcal{C}_0 \times \mathbb{R}^n$. The quadratic form associated with the expected Fisher information matrix $\mathcal{I}(\mathbf{c}, \bm{\beta})$ evaluates to
\begin{equation}
  \mathbf{u}^{\top} \mathcal{I}(\mathbf{c}, \bm{\beta}) \mathbf{u}
  \;=\;
  \sum_{(i,j) \in \mathcal{D}} \int_0^T \Bigl( \nabla \log \lambda_{ij}(t)^{\top} \mathbf{u} \Bigr)^2 \lambda_{ij}(t) \, dt.
\end{equation}
Because intensity functions satisfy $\lambda_{ij}(t) > 0$ everywhere on $[0,T]$, this quadratic form equals zero if and only if the directional derivative vanishes identically across time, meaning
\begin{equation}
\label{eq:directional-zero}
  \nabla \log \lambda_{ij}(t)^{\top} \mathbf{u}
  \;=\;
  0
  \quad \text{for all } (i,j) \in \mathcal{D} \text{ and } t \in [0,T].
\end{equation}
Proving that \eqref{eq:directional-zero} forces $\mathbf{u} = \mathbf{0}$ simultaneously establishes that the differential of the map $(\mathbf{c}, \bm{\beta}) \mapsto \{ \log \lambda_{ij}(\cdot) \}$ is injective and that $\mathcal{I}(\mathbf{c}, \bm{\beta})$ is strictly positive definite.

Let $\mathbf{B}(t) = (B_1(t), \dots, B_K(t))^{\top} \in \mathbb{R}^K$ denote the vector of $K$ B-spline basis functions evaluated at time $t$. The trajectory perturbation for node $i$ driven by coefficient vector $\mathbf{u}_{\mathbf{c}_i}$ is $h_i(t) = (\mathbf{B}(t)^{\top} \otimes \mathbf{I}_2) \mathbf{u}_{\mathbf{c}_i}$. Expanding \eqref{eq:directional-zero} yields
\begin{equation}
\label{eq:diff-expanded}
  u_{\beta_i} + u_{\beta_j}
  \;-\;
  2 \bigl( z_i(t) - z_j(t) \bigr)^{\top} \bigl( h_i(t) - h_j(t) \bigr)
  \;=\;
  0.
\end{equation}

Under the gauge constraints defining $\mathcal{C}_0$, node 1 is fixed at the spatial origin, imposing $z_1(t) \equiv \mathbf{0}$ and $h_1(t) \equiv \mathbf{0}$. For any dyad $(1, j)$ with $j \ge 2$, equation \eqref{eq:diff-expanded} simplifies to
\begin{equation}
\label{eq:node1-dyad}
  u_{\beta_1} + u_{\beta_j} \;-\; 2 z_j(t)^{\top} h_j(t) \;=\; 0,
\end{equation}
showing that $2 z_j(t)^{\top} h_j(t) = u_{\beta_1} + u_{\beta_j}$ is constant over $t \in [0,T]$.

For any dyad $(i,j)$ with $i, j \ge 2$, expanding \eqref{eq:diff-expanded} and substituting \eqref{eq:node1-dyad} gives
\begin{equation}
\label{eq:cross-terms}
  2 z_i(t)^{\top} h_j(t) + 2 z_j(t)^{\top} h_i(t) \;=\; 2 u_{\beta_1}.
\end{equation}


Under assumptions (A1), (A2), and (A3) with $n \ge 7$, the trajectories $\{z_i(t)\}_{i=1}^n$ span two spatial dimensions non-degenerately over time. The time-varying inner products in \eqref{eq:cross-terms} cannot sum to a temporal constant across all dyadic pairings unless $h_i(t)$ corresponds to an infinitesimal rigid spatial transformation $h_i(t) = \mathbf{R} z_i(t)$ for a skew-symmetric matrix $\mathbf{R}$. Under the rotational gauge constraint of $\mathcal{C}_0$ restricting node 2 (e.g., $h_{2y}(t) \equiv 0$), we obtain $\mathbf{R} = \mathbf{0}$, forcing $h_i(t) \equiv \mathbf{0}$ for every node $i$.

Setting $h_i(t) \equiv \mathbf{0}$ reduces \eqref{eq:node1-dyad} to $u_{\beta_1} + u_{\beta_j} = 0$ for all $j \ge 2$. Substituting this relation into $u_{\beta_i} + u_{\beta_j} = 0$ for $i, j \ge 2$ forces $u_{\beta_1} = 0$, which immediately implies $u_{\beta_i} = 0$ for all $i \in V$.

Finally, because $h_i(t) = (\mathbf{B}(t)^{\top} \otimes \mathbf{I}_2) \mathbf{u}_{\mathbf{c}_i} \equiv \mathbf{0}$ and the B-spline basis functions are linearly independent, we obtain $\mathbf{u}_{\mathbf{c}_i} = \mathbf{0}$ for all $i$. Thus $\mathbf{u} = \mathbf{0}$ is the unique kernel element, confirming that the differential is injective and $\mathcal{I}(\mathbf{c}, \bm{\beta}) \succ \mathbf{0}$.
\end{proof}


\section{Node activity updates in reduced models}
\label{supplements:beta}

Under the homogeneous baseline (3) the score is linear in $e^{\beta}$ and the update is closed-form,
\begin{equation}
\label{eq:beta-homog}
  \hat{\beta}(\mathbf{c})
  \;=\;
  \log m
  \;-\;
  \log \sum_{(i,j) \in \mathcal{D}} I_{ij},
\end{equation}
while under the Poisson $\beta$-model (4) we have $I_{ij} = T$, so
we can iteratively update
\begin{equation}
\label{eq:beta-activity}
  \beta_i^{(s+1)}
  \;=\;
  \log m_i - \log T
  \;-\;
  \log \sum_{j \ne i} \exp\bigl(\beta_j^{(s)}\bigr),
\end{equation}
a coordinate-ascent recursion of the same form as classical Bradley--Terry scaling.

\section{Details on BIC computation}
\label{supp:bic}

\subsection{Gauss--Newton curvature substitution}

The log-likelihood $\ell(\mathbf{c}, \bm{\beta})$ depends on trajectory coefficients $\mathbf{c}$ nonlinearly through pairwise Euclidean distances $d_{ij}(t) = \|\mathbf{z}_i(t) - \mathbf{z}_j(t)\|$. Standard Laplace approximations expand $\ell$ quadratically using the observed negative Hessian $\mathbf{H}_{\ell} = -\nabla_{\mathbf{c}}^2 \ell$. However, because the distance mapping is non-convex, $\mathbf{H}_{\ell}$ can possess negative eigenvalues away from local maxima.

We resolve this instability by replacing $\mathbf{H}_{\ell}$ with the positive-semidefinite Gauss--Newton curvature matrix $\mathbf{F} = \mathbf{J}^{\top}\mathbf{W}\mathbf{J}$, where $\mathbf{J}$ is the Jacobian of predicted dyadic intensities with respect to $\mathbf{c}$, and $\mathbf{W}$ is the diagonal weight matrix derived from model expectations. This substitution omits only second-derivative terms of the distance mapping responsible for negative curvature, isolating information provided by the data.

The penalized curvature operator $\mathbf{A} = \mathbf{F} + \mathbf{H}_{\mathcal{P}}$ combines data curvature with the penalty matrix $\mathbf{H}_{\mathcal{P}} = \mathbf{I}_{2n-3} \otimes \bm{\Pi}$, where $\bm{\Pi} = \tau^{-2}\mathbf{I}_K + 2\rho_1\bm{\Omega}_1 + 2\rho_2\bm{\Omega}_2 \succ \mathbf{0}$. Because $\bm{\Pi}$ is strictly positive definite, $\mathbf{A}$ satisfies $\mathbf{A} \succeq \tau^{-2}\mathbf{I}_{d_{\mathrm{free}}}$ with $d_{\mathrm{free}} = (2n-3)K$. Consequently, $\mathbf{A}$ remains strictly invertible and well-conditioned across all grid locations, ensuring that $k_{\mathrm{eff}} \in (0, d_{\mathrm{free}}]$ is rigorously defined.

\subsection{Matrix-free Hutchinson trace estimation}

Direct computation of (27) via dense matrix inversion costs $O(d_{\mathrm{free}}^3)$ operations, which becomes intractable when $d_{\mathrm{free}}$ is large. We avoid explicit construction of $\mathbf{A}$ by using the complementary trace identity
\begin{equation}
\label{eq:supp-hutch}
  k_{\mathrm{eff}}
  \;=\;
  d_{\mathrm{free}} \,-\, tr \bigl(\mathbf{H}_{\mathcal{P}}\mathbf{A}^{-1}\bigr)
  \;\approx\;
  d_{\mathrm{free}} \,-\, \frac{1}{P}\sum_{p=1}^{P} \mathbf{u}_p^{\top}\mathbf{H}_{\mathcal{P}}\,\mathbf{x}_p,
\end{equation}
where $\mathbf{u}_1, \dots, \mathbf{u}_P \in \{-1,+1\}^{d_{\mathrm{free}}}$ are independent Rademacher probe vectors, and each $\mathbf{x}_p$ is the unique solution to the linear system $\mathbf{A}\mathbf{x}_p = \mathbf{u}_p$.

Matrix-vector products $\mathbf{A}\mathbf{v} = \mathbf{F}\mathbf{v} + \mathbf{H}_{\mathcal{P}}\mathbf{v}$ are computed without forming dense matrices. The action $\mathbf{F}\mathbf{v}$ is evaluated via two-pass automatic differentiation, while $\mathbf{H}_{\mathcal{P}}\mathbf{v}$ exploits the block-Kronecker structure of $\mathbf{H}_{\mathcal{P}}$, requiring only banded matrix operations involving $\bm{\Pi}$.

Each system $\mathbf{A}\mathbf{x}_p = \mathbf{u}_p$ is solved using a preconditioned conjugate gradient solver. To ensure smooth output across the $(\rho_1,\rho_2)$ evaluation grid, we fix the random seeds generating probe vectors $\{\mathbf{u}_p\}_{p=1}^P$ across all grid cells and warm-start conjugate gradient iterations using solutions from adjacent grid points. This eliminates Monte Carlo sampling noise, producing a continuous criterion surface.

\subsection{Optimization surface characteristics}

Empirically, the criterion surface displays distinct sensitivities along the two tuning axes. The surface responds strongly to variations in $\rho_2$, which governs path curvature and balances model flexibility against spatial smoothness. Sensitivity to $\rho_1$ is comparatively muted, as $\rho_1$ primarily scales total path length. To maximize computational efficiency, we recommend evaluating a broad preliminary grid over $(\rho_1, \rho_2)$, followed by a refined unidimensional search over $\rho_2$ at a fixed moderate value of $\rho_1$.

\section{Details of simulations and application}
\label{supplement:simulations}
 
This section includes the additional details for the simulation studies of Section~5 and application in Section~6, covering the evaluation metrics, the two extrapolation rules used out of sample, and the implementation of each competing method.

\paragraph{Metrics.}
The distance IMSE is $\{T\binom{n}{2}\}^{-1}\sum_{i<j}\int_0^T\{\hat d_{ij}(t)-d_{ij}(t)\}^2dt$, which is invariant to rotation, reflection and translation of either configuration and so needs no Procrustes alignment. The held-out log-likelihood is $\sum_{e}\log\hat\lambda_{i_ej_e}(t_e)-\sum_{i\neq j}\int_{T_{\mathrm{tr}}}^{T}\hat\lambda_{ij}$, the sum running over held-out events and the integral evaluated by Gauss--Legendre quadrature. 
CLPM and the dyad-independent Poisson process pool the two directions of each dyad into a single process, so their fitted intensity estimates twice the per-ordered-pair rate and every quantity reported for them uses $\hat f_{ij}(t) - \log 2$. Without that adjustment every predicted count would be doubled.
The count metric compares that integral with the observed held-out count on each ordered pair. The Kolmogorov--Smirnov statistic rescales each ordered pair's event times by its own compensator, completing the right-censored final interval as its residual plus an independent unit exponential draw, and comparing the pooled intervals with the uniform distribution. Since the parameters are estimated from the same events and $m$ is large, it is reported as a comparative discrepancy across specifications rather than as a test.

\paragraph{Extrapolation rules.}
Persistence freezes the latent positions, or the log-intensity where there are no positions, at $T_{\mathrm{tr}}$. Constant velocity continues them at the analytic spline derivative for DLS-PP and the homogeneous baseline, at the slope of the segment containing $T_{\mathrm{tr}}$ for CLPM, whose piece-wise linear paths have no derivative at a change point, and at the fitted log-intensity derivative for the dyad-independent Poisson process. The Poisson $\beta$-model has no trajectory, so the two rules coincide and one held-out block is reported, an identity we verify numerically rather than assume.
 
\paragraph{Implementation of CLPM.}
We use the distance variant of the authors' implementation. It carries a single increment penalty, the analogue of $\rho_1$, and no analogue of $\rho_2$, since piece-wise linear paths have no curvature within a segment. The penalty is selected by BIC with the parameter count $2nK_{\mathrm{cp}}+1$, no effective dimension being available, and the number of change points is fixed so that the coefficient count per coordinate matches $K$. 
 
\paragraph{Implementation of the dyad-independent Poisson process.}
Each pooled dyad is fitted independently with an unpenalized B-spline log-intensity of dimension $K$, giving $\binom{n}{2}K$ coefficients against $(2n-3)K+n$ for DLS-PP. The log-likelihood of a dyad depends on its events only through $\sum_e\phi(t_e)$, so all dyads are fitted simultaneously by a batched Newton iteration. A dyad with no events in the fitting window has $\hat\lambda_{ij}\equiv0$, which attains its maximum contribution of exactly zero and requires no optimization, and a held-out event on such a dyad has log-intensity $-\infty$. The held-out log-likelihood is therefore accumulated over dyads with at least one training event and the number of held-out events on the remainder is reported separately, with no floor placed on $\hat\lambda$. Coefficients are bounded during optimization only to keep the arithmetic finite, since with fewer events than coefficients on a dyad the maximum is approached but not attained. Because the endpoint derivative of such a fit is large and is exponentiated, the constant-velocity cell is reported as unstable rather than as a figure, and the distance IMSE is replaced by $\mathrm{IMSE}_\lambda$ on the conditional intensities.
 
\paragraph{Implementation of the Poisson $\beta$-model.}
This specification has no tuning parameters. As shown in Section~\ref{supplements:beta}, given $m_i$ events involving node $i$, the profile score is $m_i 2T\exp(\beta_i)\sum_{j\neq i}\exp(\beta_j)$ and we solve it by the half-stepped iteration $\beta_i\leftarrow\tfrac12[\beta_i+\log m_i-\log\{2T\sum_{j\neq i}\exp(\beta_j)\}]$. The half step is necessary because the undamped map has eigenvalue $-1$ along the uniform direction, so its overall level oscillates with period two while the
deviations from the mean converge.
 
\paragraph{Replication and tuning.}
Each design uses a ground truth fixed across replications, with one dataset drawn per replication and the optimizer seeded by that replication's data seed, so the reported standard deviations cover both sources of variation. Tuning is selected once by BIC for DLS-PP, homogeneous baseline and CLPM, and then held fixed. The error in $\boldsymbol\beta$ is undefined for the dyad-independent Poisson process, and for CLPM and the homogeneous baseline it is reported on the single-intercept scale.
 
\paragraph{Selected tuning parameters.}
Table~\ref{tab:selected-tuning} records the values used. The in-sample and held-out columns differ because the two blocks come from separate fits, the first to the whole window and the second to $[0,T_{\mathrm{tr}}]$ only, each with its tuning selected on the events that fit saw. The basis dimension is $K=10$ throughout for models with B-spline representations including DLS-PP, homogeneous baseline, and DIPP, and for CLPM the number of change points is selected and fixed at $9$.

\begin{table}[htbp]
\centering
\caption{Tuning parameters used for each method, in the two simulation designs and in the application. The dyad-independent Poisson process carries nothing beyond the basis dimension per dyad and the Poisson $\beta$-model carries nothing at all.}
\label{tab:selected-tuning}
\small
\setlength{\tabcolsep}{4pt}
\begin{tabular}{@{}llcccccc@{}}
\toprule
& & \multicolumn{2}{c}{S1} & \multicolumn{2}{c}{S2}
  & \multicolumn{2}{c}{Application} \\
\cmidrule(lr){3-4}\cmidrule(lr){5-6}\cmidrule(lr){7-8}
Method & Parameter & \makecell{In-\\sample} & \makecell{Held-\\out} & \makecell{In-\\sample} & \makecell{Held-\\out}
        & \makecell{In-\\sample} & \makecell{Held-\\out} \\
\midrule
\multirow{2}{*}{DLS-PP}
  & $\rho_1$ & 0.1 & 0.01 & 0.1 & 0.1 & 0.1 & 0.0001 \\
  & $\rho_2$ & 0.005 & 0.01 & 100 & 100 & 0.0001 & 0.0001 \\
\addlinespace
\multirow{2}{*}{\makecell[l]{Homogeneous\\baseline}}
  & $\rho_1$ & 1 & 0.01 & 0.1 & 0.1 & 0.0001 & 0.0001 \\
  & $\rho_2$ & 0.001 & 0.01 & 100 & 100 & 0.0001 & 0.0001 \\
\addlinespace
CLPM    & penalty       & 1 & 0.1 & 0.1 & 0.001 & 1 & 1 \\
\addlinespace
DIPP            & $K$ & \multicolumn{6}{c}{10} \\
\addlinespace
Poisson-$\beta$ & ---          & \multicolumn{6}{c}{none} \\
\bottomrule
\end{tabular}
\end{table}

\section{60 selected countries in application}


To select the cohort for Section 6, economic size was measured using constant purchasing power parity (PPP) GDP from the World Bank World Development Indicators (series \texttt{NY.GDP.MKTP.PP.KD}), which adjusts for inflation and local price levels to ensure real comparability across time. We excluded regional aggregates and retained sovereign entities with at least 23 years of non-missing GDP and population data over 1995--2022. Eligible economies were ranked annually by descending PPP GDP. The 60 countries with the lowest median annual rank across the study period form the final panel. Table~\ref{tab:top60_gdp_countries} lists the selected economies and their rankings.

\begin{table}[htbp]
\centering
\tiny
\caption{Selected 60 economies (index order), 1995--2022.}
\label{tab:top60_gdp_countries}
\begin{tabular}{r l r l r l r l}
\hline
\# & Country & \# & Country & \# & Country & \# & Country \\
\hline
1  & United States            & 16 & Turkiye                  & 31 & Colombia                 & 46 & Chile \\
2  & China                    & 17 & Saudi Arabia             & 32 & Viet Nam                 & 47 & Iraq \\
3  & India                    & 18 & Nigeria                  & 33 & Switzerland              & 48 & Portugal \\
4  & Japan                    & 19 & Australia                & 34 & Bangladesh               & 49 & Hong Kong SAR, China \\
5  & Russian Federation       & 20 & Thailand                 & 35 & Philippines              & 50 & Denmark \\
6  & Germany                  & 21 & Egypt, Arab Rep.         & 36 & Sweden                   & 51 & Peru \\
7  & Brazil                   & 22 & Iran, Islamic Rep.       & 37 & Algeria                  & 52 & Finland \\
8  & France                   & 23 & Argentina                & 38 & Austria                  & 53 & Hungary \\
9  & United Kingdom           & 24 & Netherlands              & 39 & United Arab Emirates     & 54 & Ireland \\
10 & Italy                    & 25 & Poland                   & 40 & Greece                   & 55 & Israel \\
11 & Mexico                   & 26 & Pakistan                 & 41 & Romania                  & 56 & Morocco \\
12 & Indonesia                & 27 & Ukraine                  & 42 & Singapore                & 57 & Angola \\
13 & Spain                    & 28 & South Africa             & 43 & Kazakhstan               & 58 & Syrian Arab Republic \\
14 & Canada                   & 29 & Malaysia                 & 44 & Norway                   & 59 & Sri Lanka \\
15 & Korea, Rep.              & 30 & Belgium                  & 45 & Czechia                  & 60 & Kuwait \\
\hline
\end{tabular}
\end{table}

\end{document}

%% file: preamble.tex
\usepackage[utf8]{inputenc} 
\usepackage[T1]{fontenc}    

\usepackage[bitstream-charter]{mathdesign}
\usepackage{amsmath}
\usepackage[scaled=0.92]{PTSans}

\usepackage[
  paper  = letterpaper,
  left   = 1.0in,
  right  = 1.0in,
  top    = 1.0in,
  bottom = 1.0in,
  ]{geometry}

\usepackage[usenames,dvipsnames,table]{xcolor}
\definecolor{shadecolor}{gray}{0.9}

\usepackage[final,expansion=alltext]{microtype}
\usepackage[english]{babel}
\usepackage[parfill]{parskip}
\usepackage{afterpage}
\usepackage{framed}
\newcommand{\spacingset}[1]{%
  \renewcommand{\baselinestretch}{#1}\small\normalsize}
\spacingset{1.2}

{\endMakeFramed}

\usepackage{lineno}

\usepackage{ragged2e}

\newcounter{parcount}

\usepackage{graphicx}
\usepackage{wrapfig}
\usepackage[labelfont=bf,font={footnotesize,stretch=1},width=.9\textwidth]{caption}
\usepackage[format=hang]{subcaption}

\usepackage{booktabs,multirow,multicol}       

\usepackage{algorithm}
\usepackage{algpseudocode}
\usepackage{listings}
\usepackage{fancyvrb}
\fvset{fontsize=\normalsize}

\usepackage[colorlinks,linktoc=all]{hyperref}
\usepackage[all]{hypcap}
\hypersetup{citecolor=MidnightBlue}
\hypersetup{linkcolor=MidnightBlue}
\hypersetup{urlcolor=MidnightBlue}

\usepackage[acronym,nowarn]{glossaries}

\lstdefinestyle{mystyle}{
    commentstyle=\color{OliveGreen},
    keywordstyle=\color{BurntOrange},
    numberstyle=\tiny\color{black!60},
    stringstyle=\color{MidnightBlue},
    basicstyle=\ttfamily,
    breakatwhitespace=false,
    breaklines=true,
    captionpos=b,
    keepspaces=true,
    numbers=left,
    numbersep=5pt,
    showspaces=false,
    showstringspaces=false,
    showtabs=false,
    tabsize=2
}
\input{preamble/preamble_math}
\input{preamble/definitions_basic}

\input{preamble/commenting}

\usepackage{xr-hyper}
\usepackage{makecell}

\AtBeginEnvironment{abstract}{%
  \renewcommand{\baselinestretch}{1.2}\selectfont}

\usepackage{booktabs,multirow,multicol}       
\usepackage{makecell}

%% file: preamble/preamble_math.tex
\usepackage{centernot}
\usepackage{amsthm}         
\usepackage{amsmath,bm,bbm}
\usepackage{nicefrac}       
\usepackage{mathtools}      
\usepackage{amsbsy}         
\usepackage{amstext}        
\usepackage{thmtools}       
\usepackage{thm-restate}    
\usepackage{algorithm,algpseudocode} 

\begingroup
    \makeatletter
    \@for\theoremstyle:=definition,remark,plain\do{%
        \expandafter\g@addto@macro\csname th@\theoremstyle\endcsname{%
            \addtolength\thm@preskip\parskip
            }%
        }
\endgroup

\usepackage{arydshln}
\makeatletter
\def\adl@drawiv#1#2#3{%
        \hskip.5\tabcolsep
        \xleaders#3{#2.5\@tempdimb #1{1}#2.5\@tempdimb}%
                #2\z@ plus1fil minus1fil\relax
        \hskip.5\tabcolsep}
\newcommand{\cdashlinelr}[1]{%
  \noalign{\vskip\aboverulesep
           \global\let\@dashdrawstore\adl@draw
           \global\let\adl@draw\adl@drawiv}
  \cdashline{#1}
  \noalign{\global\let\adl@draw\@dashdrawstore
           \vskip\belowrulesep}}
\makeatother

\renewcommand{\epsilon}{\varepsilon}

\declaretheorem[style=plain,numberwithin=section,name=Proposition]{proposition}

\declaretheorem[style=definition,numberwithin=section,name=Example]{example}
\declaretheorem[style=remark,numberwithin=section,name=Remark]{remark}

\DeclarePairedDelimiter{\norm}{\lVert}{\rVert}

\newenvironment{example*}
 {\pushQED{\qed}\example}
 {\popQED\endexample}
\numberwithin{equation}{section}
\newcommand{\dlspp}{DLS-PP}
\DeclareMathOperator{\tr}{tr}

%% file: preamble/definitions_basic.tex
\DeclareMathOperator*{\argmax}{argmax}


%% file: preamble/commenting.tex
\definecolor{WowColor}{rgb}{.75,0,.75}
\definecolor{SubtleColor}{rgb}{0,0,.50}

\newcounter{margincounter}

%% file: tex_files/0_abstract.tex

We study continuous-time relational event data, where time-stamped dyadic interactions reflect both individual node propensities and evolving relational proximity. We propose a dynamic latent space model for inhomogeneous Poisson processes, where event intensities depend on node-specific activity parameters and time-varying latent distances modeled via flexible B-splines. We prove model identifiability by decoupling baseline activity from latent position, ensuring high interaction volumes do not warp the spatial map. For scalability, we develop a minibatch stochastic gradient algorithm with stable initialization and geometric anchoring, alongside an effective-degrees-of-freedom BIC for tuning model complexity. Simulations confirm accurate parameter recovery and out-of-sample prediction. Applied to cooperative diplomatic events among 60 major economies (1995--2022), the model uncovers shifting patterns of international cooperation and isolates mobile geopolitical actors from stationary institutional anchors.

%% file: tex_files/1_introduction.tex
\section{Introduction}
\label{sec:intro}

Understanding how entity relationships evolve over time is a fundamental goal across the social, political, and biological sciences. Latent space models \citep[LSMs;][]{hoff_2002_latent} offer a powerful geometric framework for this task, embedding entities in a low-dimensional space where proximity represents relational affinity. However, modern relational data rarely arrive as static networks. Instead, they consist of continuous streams of individual event records that precisely mark when two entities connect.

\begin{figure}[htbp]
    \centering
    \includegraphics[width=\linewidth]{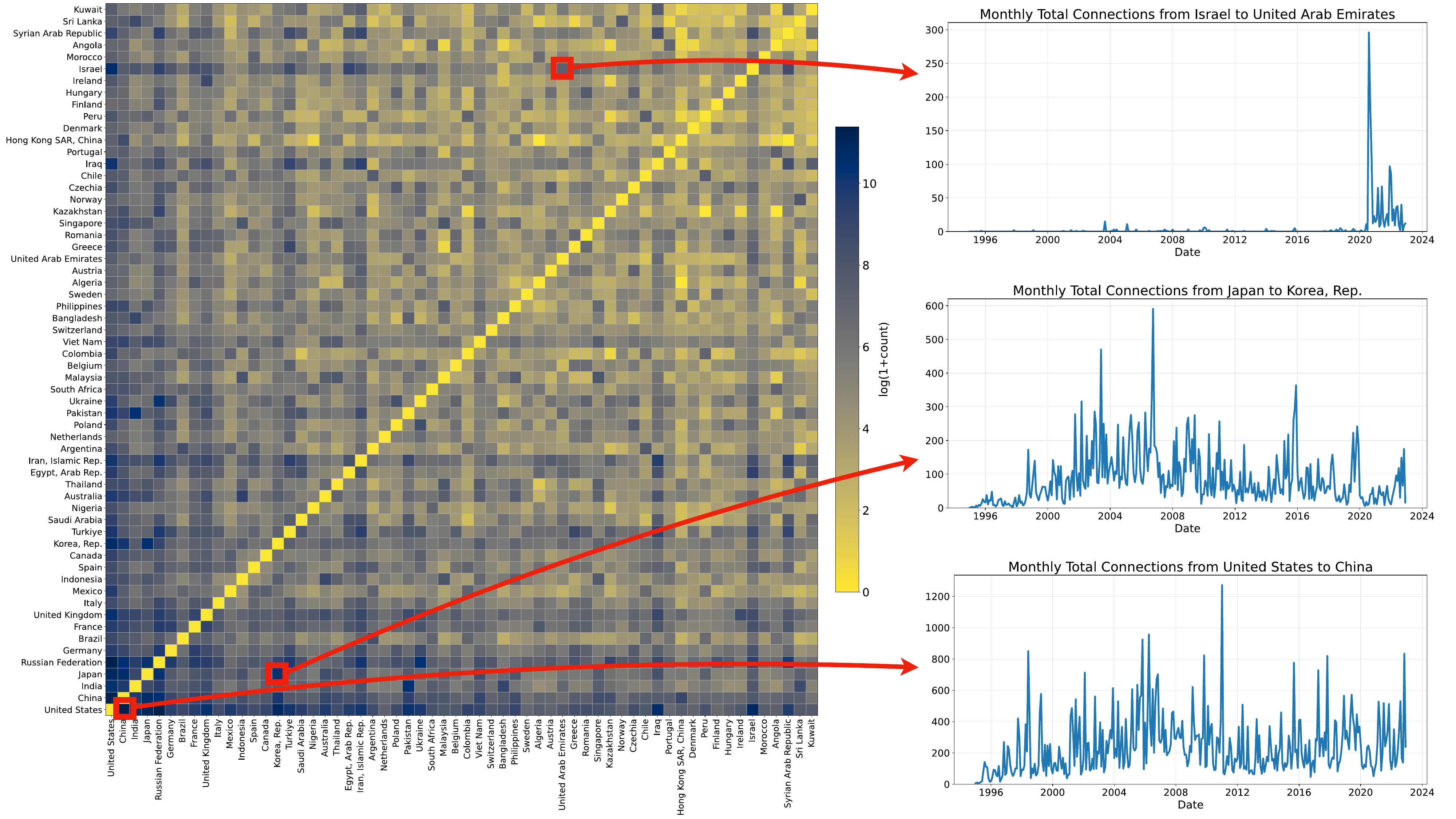}
  \caption{Dyad-level ICEWS cooperative interactions. The left panel reports total events between countries $i$ and $j$ over 1995--2022 across 60 economies. The right panel shows monthly counts for three dyads with comparable totals (Israel to UAE, Japan to Korea, and US to China), whose contrasting dynamics aggregate totals conceal.}
    \label{fig:icews_heatmap_total}
\end{figure}

A motivating example is the Integrated Crisis Early Warning System \citep[ICEWS;][]{Boschee2015ICEWS}, a machine-coded archive that records directed political actions capturing who did what to whom and when between pairs of nations over decades, and that has been externally validated against human coding \citep{BagozziEtAl2019Underreporting}.
Such data invite questions about how dyadic alignments drift, how networks realign, and whether global shocks alter long-term cooperation. Yet, collapsing these event streams into static summary counts discards the temporal dynamics needed to answer them. As Figure~\ref{fig:icews_heatmap_total} demonstrates, dyads with similar aggregate totals can mask fundamentally distinct interaction histories, ranging from abrupt onsets and recurrent bursts to gradual drift.

Many dynamic LSMs attempt to track temporal changes by fitting separate latent geometries across sequences of discrete snapshots \citep{sarkar2005dynamic, sewell2015latent, durante2014nonparametric, loyal2024fast,macdonald2025latent}.
Extensions like \citet{durante2016locally} allow latent positions to evolve in continuous time via stochastic differential equations, but their likelihoods remain bound to time-aggregated network snapshots rather than point-process event streams. Alternatively, some approaches fix spatial geometries entirely while allowing only baseline node activity to vary \citep{he2025semiparametric}. Consequently, these snapshot-based frameworks lack an instantaneous event-generating mechanism, making it difficult to capture fine-grained temporal dynamics or perform smooth forward-in-time forecasting. 

This paper develops a continuous-time dynamic latent space model built directly on point-process event streams. By representing latent positions as continuous curves, our framework models trajectory paths and instantaneous velocities while avoiding temporal aggregation. We establish conditions under which time-varying latent geometry and node-level baseline activity are separately identified, and we demonstrate the model's predictive accuracy and interpretability by analyzing cooperative political events among 60 major economies from 1995 to 2022.

Existing methods that model continuous event timestamps resolve parts of this challenge, but key gaps remain. Continuous-time stochastic blockmodels \citep{corneli2018multiple, matias2018semiparametric, xin2017continuous} capture discrete community shifts rather than continuous spatial geometries, whereas latent Hawkes models \citep{huang_2022_mutually, ward_2022_network} accommodate temporal dynamics while keeping latent positions static. Most relevant to our work are continuous latent position models \citep{artico2023dynamic, rastelli_2023_continuous, romero_2023_gaussian}, which allow positions to move over time. However, current formulations suffer from major limitations: they constrain trajectories to piecewise-linear paths with unnatural directional kinks, omit node-specific baseline activity parameters—forcing overall interaction volume to be confounded with spatial distance, and face severe computational bottlenecks when evaluating point-process likelihoods over large networks.

To address these challenges, we introduce a continuous-time dynamic latent space model built directly on relational event streams. We model interactions between nodes $i$ and $j$ as independent inhomogeneous Poisson processes with intensity
\begin{equation}\label{eq:intro-model}
  \log \lambda_{ij}(t) \;=\; \beta_i + \beta_j - \|z_i(t) - z_j(t)\|^2 ,
  \qquad 1 \le i < j \le n ,
\end{equation}
where $z_i(t) \in \mathbb{R}^2$ denotes a smooth latent trajectory and $\beta_i$ captures node $i$'s intrinsic activity. In diplomatic networks, omitting $\beta_i$ artificially pulls "busy" high-volume actors toward the geometric center, confounding baseline activity with political alignment \citep{krivitsky2009representing, rastelli_2023_continuous}. Our formulation decouples activity from spatial geometry so that positions reflect true relational affinity. Crucially, we prove that this joint specification is fully identifiable without ad hoc constraints, yielding a smooth continuous-time framework that avoids temporal aggregation artifacts while scaling to large event archives.

This paper provides the following four key contributions.

\begin{enumerate}
\item \textbf{Continuous-time smooth embeddings with identified nodal effects.} We formulate latent trajectories as smooth, continuous curves governed by flexible roughness penalties rather than piecewise approximations. Notably, additive nodal parameters decouple baseline activity heterogeneity from spatial distance, making intrinsic interaction propensity and relational geometry separately inferable components.
 
\item \textbf{Theoretical identifiability of smooth geometry and activity.} We provide formal identifiability proofs for both the continuous-time geometry and node-level activity parameters. Our anchoring framework resolves continuous-time rotational and translational invariance, while strict concavity guarantees unique activity estimates. 

\item \textbf{Scalable computation and model selection.} We overcome the severe optimization bottlenecks of dynamic LSMs, enabling inference on large-scale datasets via dyadic minibatch stochastic gradient descent, while delivering a BIC criterion powered by generalized Gauss-Newton optimization and Hutchinson trace estimation of effective degrees of freedom.
 
 
\item \textbf{Evidence on the changing geometry of international cooperation.} Applied to cooperative ICEWS events among 60 major economies over 1995--2022, the proposed method yields excellent predictive performance and latent trajectories of nations, and the nested specifications let us test whether both the geometry and the activity parameters are needed.
\end{enumerate}

The rest of the paper is organized as follows: Section~\ref{sec:model} presents the model framework, Section~\ref{sec:ident} establishes theoretical identifiability, Section~\ref{sec:estimation} outlines the estimation pipeline, Section~\ref{sec:simulation} reports simulations, and Section~\ref{sec:app} analyzes the ICEWS data. Section~\ref{sec:discussion} concludes, with technical proofs deferred to the Supplementary Material, and code reproducing all analyses is provided with the submission.

%% file: tex_files/2_methodology.tex

%

\section{A dynamic latent space model for network event data}
\label{sec:model}

We model the observed interactions as a collection of dyadic point processes whose intensities are governed by the time-varying positions of the nodes in a low-dimensional latent space. Section~\ref{subsec:data} fixes notation and the observation scheme, Section~\ref{subsec:intensity} introduces the intensity model and the two reduced specifications against which it is compared. Section~\ref{subsec:splines} gives the functional representation of the latent trajectories, Section~\ref{subsec:penalty} the penalty, and Section~\ref{subsec:anchoring} the anchoring method, objective and the estimator.

\subsection{Data structure and notation}
\label{subsec:data}

Let $V=\{1,\dots,n\}$ be a fixed set of nodes observed over the window $[0,T]$, and let
\[
  \mathcal{D} \;=\; \bigl\{(i,j) : 1 \le i < j \le n \bigr\},
  \qquad |\mathcal{D}| \;=\; \tbinom{n}{2},
\]
denote the set of unordered dyads. For each dyad we observe the times at which an interaction between $i$ and $j$ is recorded,
\[
  \mathcal{H}_{ij} \;=\; \bigl\{ t_{ij,1} < \cdots < t_{ij,m_{ij}} \bigr\}
  \;\subset\; [0,T],
\]
pooling events in both directions, with $m_{ij} = |\mathcal{H}_{ij}|$. We write $m_i = \sum_{j \ne i} m_{ij}$ for the number of events involving node~$i$ and $m = \sum_{(i,j) \in \mathcal{D}} m_{ij}$ for the total number of events. Without loss of generality, time is rescaled so that $T=1$, and we retain $T$ in the notation to keep the role of the observation length explicit. The node set is fixed over $[0,T]$ and the observation window is common to all dyads.
Treating dyads as unordered is a modelling choice rather than a property of the data. We empirically support this symmetry in Section~\ref{sec:app}.


\subsection{Intensity model}
\label{subsec:intensity}

Each node $i$ is assigned a latent trajectory
\[
  z_i(t) \;=\; \bigl(x_i(t),\, y_i(t)\bigr)^{\!\top} \;\in\; \mathbb{R}^{2},
  \qquad t \in [0,T],
\]
and we write $d_{ij}(t) = \norm{z_i(t) - z_j(t)}$ for the latent distance between $i$ and $j$ at time~$t$. Conditionally on the trajectories $\mathbf{Z}(\cdot) = \{z_1(\cdot),\dots,z_n(\cdot)\}$ and on activity parameters $\bm{\beta} = (\beta_1,\dots,\beta_n)^{\!\top}$, the dyadic processes are taken to be independent inhomogeneous Poisson processes with log intensities
\begin{equation}
\label{eq:intensity}
  \log \lambda_{ij}(t)
  \;=\; \beta_i + \beta_j - d_{ij}(t)^2,
  \qquad (i,j) \in \mathcal{D}.
\end{equation}
We refer to \eqref{eq:intensity} as the 
Dynamic Latent Space Poisson Process
(\emph{\dlspp{}}) model. The parameter $\beta_i$ captures the overall propensity of node~$i$ to interact, absorbing the degree heterogeneity.
We adopt squared distance in \eqref{eq:intensity} rather than raw Euclidean distance for computational tractability and differentiability, which ensures that the log-intensity function $\log \lambda_{ij}(t)$ remains quadratic in the latent coordinates.

Two reduced specifications serve as benchmarks throughout. Both are nested in \eqref{eq:intensity}, and we record them because the comparison in Section~\ref{sec:simulation} and~\ref{sec:app} demonstrates that both the latent geometry and the activity parameters are needed.
\begin{align}
  \text{\emph{Homogeneous baseline}:} \quad
    &\log \lambda_{ij}(t) = \beta - d_{ij}(t)^2,
    \label{eq:homog}\\
  \text{\emph{Poisson $\beta$-model}:} \quad
    &\log \lambda_{ij}(t) = \beta_i + \beta_j.
    \label{eq:activity}
\end{align}
Specification~\eqref{eq:homog} is the continuous-time analogue of the classical latent distance model \citep{hoff_2002_latent} with a single intercept.
Specification~\eqref{eq:activity} discards the geometry and is the Poisson form of the $\beta$-model for undirected networks \citep{holland1981exponential, chatterjee2011random}.

Conditional independence gives the log-likelihood
\begin{equation}
\label{eq:loglik}
  \ell(\mathbf{Z},\bm{\beta})
  \;=\;
  \sum_{(i,j) \in \mathcal{D}}
  \Biggl\{
    \sum_{t \in \mathcal{H}_{ij}} \log \lambda_{ij}(t)
    \;-\; \int_0^T \lambda_{ij}(t)\, dt
  \Biggr\}.
\end{equation}
Expression \eqref{eq:loglik} is a function of infinitely many parameters, since the unknowns are the positions $z_i(t)$ of all $n$ nodes at every $t$ in $[0,T]$.
We now describe how we constrain this quantity to give us a tractable representation, allowing inference.
Section~\ref{subsec:splines} represents each coordinate function in a spline basis, which turns the infinite-dimensional estimation problem into a finite-dimensional one, and Section~\ref{subsec:penalty} adds penalties on velocity and acceleration. These penalties are functions of the coordinates while the likelihood is a function of the distances alone, and Section~\ref{subsec:anchoring} therefore fixes the coordinate frame through anchoring.
Figure~\ref{fig:ingredients} shows what fails when any one of the three is omitted.

\begin{figure}[t]
\centering
\setkeys{Gin}{height=2.2cm, keepaspectratio} 
\setlength{\tabcolsep}{2pt}
\begin{tabular}{ccccc}
  \includegraphics{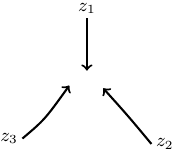} & \includegraphics{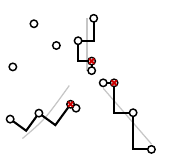} &
  \includegraphics{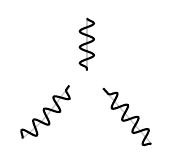} & \includegraphics{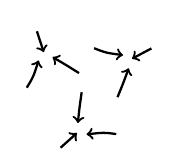} &
  \includegraphics{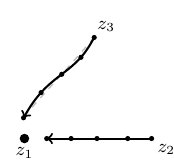} \\[0.15em]
  {\small (a)} & {\small (b)} & {\small (c)} & {\small (d)} & {\small (e)}
\end{tabular}
\caption{Trajectory estimation for three converging nodes under varying model constraints: (a) True reference configuration; (b) Pointwise estimates exhibiting temporal instability; (c) Unpenalized splines showing overfitting; (d) Unanchored rigid-motion invariance with identical objective values \eqref{eq:objective}; and (e) Smooth trajectory recovery using penalized splines under anchored parameterization \eqref{eq:C0}.}
\label{fig:ingredients}
\end{figure}

\subsection{Functional representation of the trajectories}
\label{subsec:splines}

We assume that $x_i(t)$ and $y_i(t)$ are smooth functions of time for every node $i$. Each coordinate function is represented by a linear combination of B-spline basis functions. Let $\{\phi_k\}_{k=1}^{K}$ be a clamped B-spline basis of degree $p$ on $[0,T]$ with knots $0 = \kappa_0 < \kappa_1 < \cdots < \kappa_M = T$, and set
\begin{equation}
\label{eq:spline}
  z_i(t)
  \;=\;
  \begin{pmatrix} x_i(t) \\ y_i(t) \end{pmatrix}
  \;=\;
  \begin{pmatrix}
    \sum_{k=1}^{K} c^{(x)}_{ik}\,\phi_k(t) \\[2pt]
    \sum_{k=1}^{K} c^{(y)}_{ik}\,\phi_k(t)
  \end{pmatrix},
  \qquad i = 1,\dots,n .
\end{equation}
We fix $p = 3$ with equally spaced interior knots, giving $K = M + p$ total basis functions, and arrange the corresponding coefficients as
\[
  \mathbf{c}^{(x)}_i = \bigl(c^{(x)}_{i1},\dots,c^{(x)}_{iK}\bigr)^{\!\top},
  \qquad
  \mathbf{C}_x =
  \bigl[\mathbf{c}^{(x)}_1,\dots,\mathbf{c}^{(x)}_n\bigr]^{\!\top}
  \in \mathbb{R}^{n \times K},
\]
and analogously for the second coordinate, writing $\mathbf{c} = (\mathbf{C}_x, \mathbf{C}_y)$ and $\mathbf{c}_i = (\mathbf{c}^{(x)}_i, \mathbf{c}^{(y)}_i)$. Under \eqref{eq:spline} the squared latent distance is a quadratic form in the coefficient differences,
\begin{equation}
\label{eq:dist-spline}
  d_{ij}(t)^2
  \;=\;
  \Bigl\{ \bigl(\mathbf{c}^{(x)}_i - \mathbf{c}^{(x)}_j\bigr)^{\!\top}\!\bm{\phi}(t) \Bigr\}^{2}
  +
  \Bigl\{ \bigl(\mathbf{c}^{(y)}_i - \mathbf{c}^{(y)}_j\bigr)^{\!\top}\!\bm{\phi}(t) \Bigr\}^{2},
\end{equation}
where $\bm{\phi}(t) = \bigl(\phi_1(t),\dots,\phi_K(t)\bigr)^{\!\top}$. Therefore, $\log \lambda_{ij}(t)$ is a piecewise polynomial of degree $2p$ in $t$ and the log-likelihood \eqref{eq:loglik} becomes an explicit function of $(\mathbf{c},\bm{\beta})$. We repeatedly exploit two key properties of this representation, as the basis exactly reproduces affine functions of $t$ and the compact support of $\phi_k$ ensures that all derivative Gram matrices are banded. 

\subsection{Roughness penalties and the penalized objective}
\label{subsec:penalty}

To enforce temporal smoothness and prevent unpenalized trajectories from overfitting localized event noise (Figure~\ref{fig:ingredients}c), we estimate $(\mathbf{c}, \bm{\beta})$ by maximizing a penalized log-likelihood. The roughness penalty comprises three components that regulate trajectory velocity, curvature, and overall spatial scale.


\paragraph{Path length}
The arc length of the trajectory of node~$i$,
\[
  \mathcal{A}_i
  \;=\;
  \int_0^T \sqrt{\dot x_i(t)^2 + \dot y_i(t)^2}\; dt,
\]
is neither quadratic in $\mathbf{c}_i$ nor differentiable where the integrand vanishes. We penalize instead the integrated squared velocity,
\begin{equation}\label{eq:f1}
  f_1(\mathbf{c}_i)
  \;=\;
  \int_0^T \bigl\{\dot x_i(t)^2 + \dot y_i(t)^2\bigr\}\,dt
  \;=\;
  \bigl(\mathbf{c}^{(x)}_i\bigr)^{\!\top}\bm{\Omega}_1\,\mathbf{c}^{(x)}_i
  \;+\;
  \bigl(\mathbf{c}^{(y)}_i\bigr)^{\!\top}\bm{\Omega}_1\,\mathbf{c}^{(y)}_i ,
\end{equation}
where $[\bm{\Omega}_1]_{k\ell} = \int_0^T \phi_k'(t)\,\phi_\ell'(t)\,dt$. By the Cauchy--Schwarz inequality,
\begin{equation}
\label{eq:arc-bound}
  \mathcal{A}_i^2 \;\le\; T\, f_1(\mathbf{c}_i),
\end{equation}
so penalizing $f_1$ therefore controls arc length from above. The matrix $\bm{\Omega}_1$ is symmetric, positive semidefinite and banded, with $[\bm{\Omega}_1]_{k\ell} = 0$ whenever $|k - \ell| > p$.

\paragraph{Curvature}
The second penalty controls the smoothness of the paths through the integrated
squared acceleration,
\begin{equation}
\label{eq:f2}
  f_2(\mathbf{c}_i)
  \;=\;
  \int_0^T \bigl\{ \ddot x_i(t)^2 + \ddot y_i(t)^2 \bigr\}\, dt
  \;=\;
  \bigl(\mathbf{c}^{(x)}_i\bigr)^{\!\top} \bm{\Omega}_2\, \mathbf{c}^{(x)}_i
  \;+\;
  \bigl(\mathbf{c}^{(y)}_i\bigr)^{\!\top} \bm{\Omega}_2\, \mathbf{c}^{(y)}_i,
\end{equation}
where $ [\bm{\Omega}_2]_{k\ell} = \int_0^T \phi_k''(t)\,\phi_\ell''(t)\, dt$ is again symmetric positive semidefinite and banded.
The two penalties act on complementary features: $f_1$ governs how far a node travels, $f_2$ how directly it does so. Both are of substantive interest in the application, where the length of a country's path and the abruptness of its movements admit distinct interpretations.

\paragraph{Penalized objective}
Adding a ridge term on the coefficients, the overall penalty is given by
\begin{equation}
\label{eq:penalty}
  \mathcal{P}(\mathbf{c})
  \;=\;
  \frac{1}{2\tau^2}\,\norm{\mathbf{c}}_F^2
  \;+\;
  \rho_1 \sum_{i=1}^{n} f_1(\mathbf{c}_i)
  \;+\;
  \rho_2 \sum_{i=1}^{n} f_2(\mathbf{c}_i),
\end{equation}
where $\norm{\mathbf{c}}_F^2 = \sum_{i,k}\{ (c^{(x)}_{ik})^2 + (c^{(y)}_{ik})^2 \}$.
In practice, the structured terms evaluate compactly in matrix form as $\sum_{i=1}^n f_r(\mathbf{c}_i) = \tr\bigl(\mathbf{C}_x \bm{\Omega}_r \mathbf{C}_x^{\top}\bigr) + \tr\bigl(\mathbf{C}_y \bm{\Omega}_r \mathbf{C}_y^{\top}\bigr)$ for $r = 1,2$, facilitating efficient matrix-based gradient evaluations. The coefficient-wise Hessian of $\mathcal{P}$ is block diagonal with common $K \times K$ block
\begin{equation}
\label{eq:Pi}
  \bm{\Pi}
  \;=\;
  \tau^{-2}\mathbf{I}_K
  \;+\;
  2\rho_1 \bm{\Omega}_1
  \;+\;
  2\rho_2 \bm{\Omega}_2
  \;\succ\; \mathbf{0},
\end{equation}
which is strictly positive definite for any finite $\tau^2$. Throughout, $\tau^2$ is fixed at a large value to maintain a weak ridge penalty, while $(\rho_1,\rho_2)$ are selected via the BIC in Section~\ref{subsec:bic}.

\begin{remark}[Bayesian interpretation]
\label{rem:bayes}
The penalty \eqref{eq:penalty} corresponds to the negative log-density of a proper prior on $\mathbf{c}$, rendering $\hat{\mathbf{c}}$ a maximum a posteriori estimate. The terms $f_1$ and $f_2$ represent continuous-time analogues of first- and second-order random-walk priors \citep{lang2004bayesian}, penalizing deviations from stationary and linear trajectories, respectively. Although $f_1$ and $f_2$ are individually improper, the weak ridge term guarantees overall prior properness ($\bm{\Pi} \succ \mathbf{0}$). This framework motivates the approximated BIC in Section~\ref{subsec:bic}, replacing raw parameter counts with an effective-dimension penalty.
\end{remark}


\subsection{Anchored parameterization}
\label{subsec:anchoring}

In principle, fitting the model simply requires maximizing the penalized log-likelihood $\ell - \mathcal{P}$. However, invariance under translations, rotations, and reflections prevents a unique solution without identifying constraints. 
Three configurations differing by such a rotation are shown in Figure~\ref{fig:ingredients}(d). We remove the ambiguity by anchoring two nodes,
\begin{equation}
\label{eq:C0}
  \mathcal{C}_0
  =
  \bigl\{
    \mathbf{c} :
    \mathbf{c}^{(x)}_1 = \mathbf{c}^{(y)}_1 = \bm{0},
    \ \
    \mathbf{c}^{(y)}_2 = \bm{0}
  \bigr\},
\end{equation}
so that $z_1(t) \equiv (0,0)^{\top}$ and $z_2(t)$ lies on the $x$-axis for all $t$, as in Figure~\ref{fig:ingredients}(e). The sign of $x_2(t)$ is left free, and the anchor pair is chosen by the rule of Section~\ref{subsec:init}, which keeps the two nodes well separated. The constraints in \eqref{eq:C0} are imposed by holding the corresponding entries of $\mathbf{c}$ at zero throughout the optimization, so the free parameters have dimension $d_{\mathrm{free}} = (2n-3)K$. The estimator is then
\begin{equation}
\label{eq:objective}
  (\hat{\mathbf{c}}, \hat{\bm{\beta}})
  \;=\;
  \argmax_{(\mathbf{c},\bm{\beta}) \,\in\, \mathcal{C}_0 \times \mathbb{R}^n}\;
  \mathcal{J}(\mathbf{c},\bm{\beta}),
  \qquad
  \mathcal{J}(\mathbf{c},\bm{\beta})
  \;=\;
  \ell(\mathbf{c},\bm{\beta}) - \mathcal{P}(\mathbf{c}).
\end{equation}
Defining the model directly on the anchored space $C_0$ rather than quotienting an unconstrained spline space is computationally efficient. Section~\ref{sec:ident} shows that \eqref{eq:C0} leaves only a finite ambiguity, and that the ambiguity acts trivially on every quantity we report.

%% file: tex_files/3_identifiability.tex
%
%
%
%
%

\section{Identifiability}
\label{sec:ident}

Restricting coefficients to the anchored set $\mathcal{C}_0$ in \eqref{eq:C0}, identifiability of $(\mathbf{c}, \bm{\beta})$ follows the mapping
\begin{equation}
\label{eq:chain}
  \text{law of } \{\mathcal{H}_{ij}\}
  \ \xrightarrow{\ (a)\ }\
  \{\lambda_{ij}(\cdot)\}
  \ \xrightarrow{\ (b)\ }\
  \bigl( \bm{\beta},\, \{ d_{ij}(t)^2 \} \bigr)
  \ \xrightarrow{\ (c)\ }\
  \{ z_i(t) \}
  \ \xrightarrow{\ (d)\ }\
  ( \mathbf{C}_x, \mathbf{C}_y ).
\end{equation}
Step (a) holds because continuous Poisson intensities are uniquely determined by their point-process law \citep{karr_1991_point}, and step (d) follows directly from B-spline linear independence.

The remaining steps form the theoretical core of this section. Proposition~\ref{prop:ident} addresses step (c) by showing that dynamic distance matrices uniquely recover spatial trajectories up to a time-invariant discrete reflection under $\mathcal{C}_0$. Proposition~\ref{prop:beta} establishes conditional concavity and existence bounds for $\bm{\beta}$ given node positions, which serves as a building block for Proposition~\ref{prop:joint} to resolve step (b) by jointly decoupling activity parameters from geometry. Finally, Proposition~\ref{prop:fisher} guarantees local non-degeneracy by proving that the expected Fisher information matrix is strictly positive definite on $\mathcal{C}_0 \times \mathbb{R}^n$, establishing the foundation for valid asymptotic inference.

Three conditions on the trajectories are used below.
\begin{enumerate}
\item[(A1)] The two anchor nodes in Section~\ref{subsec:anchoring} do not collide, $\inf_{t \in [0,T]} |x_2(t)| > 0$.
  \item[(A2)] The nodes are not collinear at any time,
    $\min_{t \in [0,T]} \max_{i > 2} |y_i(t)| > 0$. Under (A1) this is equivalent to $\operatorname{rank} \mathbf{Z}(t) = 2$ for every $t$ with $\mathbf{Z}(t) = [z_1(t),\dots,z_n(t)]^{\top}$.
  \item[(A3)] At some $t^{\dagger} \in [0,T]$, no subset of $n-4$ nodes lies on a common straight line in $\mathbb{R}^{2}$.
\end{enumerate}
Condition (A1) keeps the constrained coordinate system well defined, since the orientation of the $x$-axis will become arbitrary if the two nodes approach one another. Condition (A2) prevents the residual ambiguity from varying over time, and (A3) is used for Proposition~\ref{prop:joint}. The proofs of all propositions are given in Supplementary Section 2. 

\subsection{Identifiability of the trajectories}
\label{subsec:traj}

In general dynamic latent space models, distance-based likelihoods are invariant under rotation, translation, and reflection at each time point. The anchoring constraints in $\mathcal{C}_0$ resolve this continuous time-varying non-identifiability, reducing the space of isometric latent configurations to a discrete, time-invariant symmetry group.

\begin{proposition}
\label{prop:ident}
Let $\mathbf{c}, \mathbf{c}^{*} \in \mathcal{C}_0$ induce trajectories
$\mathbf{Z}(\cdot)$ and $\mathbf{Z}^{*}(\cdot)$ with
$d^{*}_{ij}(t) = d_{ij}(t)$ for all $(i,j) \in \mathcal{D}$ and all
$t \in [0,T]$, and let $\mathbf{Z}(\cdot)$ satisfy (A1).
\begin{enumerate}
  \item[(i)] For each $t$, there is a matrix $\mathbf{U}(t)$ that belongs to the Klein four-group $\mathbb{V}_4 = \{ \mathbf{I}_2, -\mathbf{I}_2, \operatorname{diag}(1,-1), \allowbreak \operatorname{diag}(-1,1) \}$, such that $\mathbf{Z}^{*}(t) = \mathbf{Z}(t) \mathbf{U}(t)$.
  \item[(ii)] If $\mathbf{Z}(\cdot)$ satisfies (A2) as well, then $\mathbf{U}(t) = \mathbf{U} \in \mathbb{V}_4 $ for all $t$.
  \item[(iii)] Conversely, each $\mathbf{U} \in \mathbb{V}_4$ acts on the
    coefficients by
    $(\mathbf{C}_x, \mathbf{C}_y) \mapsto
     (\epsilon_1 \mathbf{C}_x, \epsilon_2 \mathbf{C}_y)$
    with $\epsilon_1, \epsilon_2 \in \{-1,1\}$, maps $\mathcal{C}_0$ onto
    itself, and leaves $\ell$, $f_1$, $f_2$ and $\|\mathbf{c}\|_F$ unchanged.
    Under (A1) and (A2) the four resulting coefficient arrays are distinct.
\end{enumerate}
\end{proposition}

Every quantity reported in Sections~\ref{sec:simulation} and~\ref{sec:app} is invariant under $\mathbb{V}_4$, including the distances $d_{ij}(t)$, the intensities $\lambda_{ij}(t)$, the penalties $f_1$ and $f_2$ and the arc lengths $A_i$. 
Proposition~\ref{prop:ident} provides the theoretical basis for the initialization method presented in Section~\ref{subsec:init}.


\subsection{Identifiability of the activity parameters}
\label{subsec:ident-beta}

In contrast to network formulations with node-wise effects requiring sum-to-zero or corner constraints of $\bm{\beta}$, the activity vector $\bm{\beta}$ is identified without explicit normalization.

\begin{proposition}
\label{prop:beta}
Fix the trajectories $\mathbf{Z}(\cdot)$, let
$I_{ij} = \int_0^T \exp\{ -d_{ij}(t)^2 \} \, dt$, which is positive for every
dyad, and let $n \ge 3$.
\begin{enumerate}
  \item[(i)] The map $\bm{\beta} \mapsto \{ \lambda_{ij}(\cdot) \}_{(i,j) \in
    \mathcal{D}}$ is injective, so $\bm{\beta}$ is identified without further
    constraints.
  \item[(ii)] The profile objective
    $\bm{\beta} \mapsto \ell(\mathbf{Z}, \bm{\beta})$ is strictly concave, with
    \begin{equation}
    \label{eq:beta-hessian}
      - \mathbf{v}^{\top}
      \bigl\{ \nabla^2_{\bm{\beta}} \ell \bigr\}
      \mathbf{v}
      =
      \sum_{(i,j) \in \mathcal{D}}
      e^{\beta_i + \beta_j} I_{ij} (v_i + v_j)^2
      > 0
      \quad \text{for every } \mathbf{v} \neq \bm{0} .
    \end{equation}
  \item[(iii)] The supremum of $\ell(\mathbf{Z}, \cdot)$ is attained, and the maximizer is then unique, if and only if $m_i$, the number of events involving node $i$, satisfies
    \begin{equation}
    \label{eq:beta-existence}
      0 < m_i < m
      \quad \text{for every } i \in V .
    \end{equation}
\end{enumerate}
\end{proposition}

The quadratic form in \eqref{eq:beta-hessian} vanishes if and only if $v_i + v_j = 0$ for all observed dyads. This linear system forces $\mathbf{v} = \mathbf{0}$ provided the graph of observed dyads is connected and non-bipartite. Bipartiteness permits non-trivial alternating shifts $\beta_i \mapsto \beta_i + \alpha$ and $\beta_j \mapsto \beta_j - \alpha$. In our setup, observing all dyads yields the complete graph, which satisfies this condition for $n \ge 3$. Finally, condition \eqref{eq:beta-existence} guarantees the existence of a unique maximizer by requiring every node to participate in at least one event without monopolizing the entire process. 



\subsection{Joint identifiability}
\label{subsec:joint}

This subsection establishes step (b) of \eqref{eq:chain} by showing that $\bm{\beta}$ and the dynamic distances $\{d_{ij}(\cdot)^2\}$ are jointly identifiable, ensuring baseline activity cannot be traded off against spatial position.

Suppose two configurations $(\mathbf{c}, \bm{\beta}), (\mathbf{c}^*, \bm{\beta}^*) \in \mathcal{C}_0 \times \mathbb{R}^n$ generate identical intensities $\lambda_{ij}^*(t) = \lambda_{ij}(t)$ for all $(i,j) \in \mathcal{D}$ and $t \in [0,T]$. Setting $\bm{a} = \bm{\beta}^* - \bm{\beta}$ yields
\begin{equation}
\label{eq:additive-shift}
  d^{*}_{ij}(t)^2 = d_{ij}(t)^2 + a_i + a_j, \quad (i,j) \in \mathcal{D},\ t \in [0,T].
\end{equation} 
Let $\mathbf{D}(t) = [ d_{ij}(t)^2 ]_{i,j=1}^{n}$ denote the squared distance matrix, $\mathbf{Q} = \mathbf{I}_n - n^{-1} \bm{1} \bm{1}^{\top}$ the projection matrix onto the orthogonal complement of $\operatorname{span}(\bm{1})$, and $\mathbf{G}(t) = -\tfrac{1}{2} \mathbf{Q} \mathbf{D}(t) \mathbf{Q} = \mathbf{Q} \mathbf{Z}(t) \mathbf{Z}(t)^{\top} \mathbf{Q}$ the centered Gram matrix, so that $\operatorname{rank} \mathbf{G}(t) \le 2$ for all $t \in [0, T]$.

\begin{proposition}
\label{prop:joint}
Let $n \ge 7$ and suppose $(\mathbf{c}, \bm{\beta}), (\mathbf{c}^*, \bm{\beta}^*) \in \mathcal{C}_0 \times \mathbb{R}^n$ satisfy $\lambda^*_{ij}(t) = \lambda_{ij}(t)$ for all $(i,j) \in \mathcal{D}$ and $t \in [0,T]$. Define $\bm{a} = \bm{\beta}^* - \bm{\beta}$ and $\mathcal{S} = \{ i : a_i \neq 0 \}$.
\begin{enumerate}
  \item[(i)] The activity shift vector $\bm{a}$ satisfies
  \begin{equation}
  \label{eq:rank-relation}
    \mathbf{Q} \operatorname{diag}(\bm{a}) \mathbf{Q} = \mathbf{G}^{*}(t) - \mathbf{G}(t), \quad t \in [0,T],
  \end{equation}
  where $\mathbf{G}^*(t)$ is the centered Gram matrix of $\mathbf{Z}^*(t)$. We have $\operatorname{rank}(\mathbf{G}^*(t) - \mathbf{G}(t)) \le 4$, therefore, at most four activity parameters can differ, i.e., $|\mathcal{S}| \le 4$.

  \item[(ii)] Under (A3), $\bm{\beta}^* = \bm{\beta}$. If (A1)–(A3) hold, $(\mathbf{c}^*, \bm{\beta}^*)$ and $(\mathbf{c}, \bm{\beta})$ coincide up to an action of $\mathbb{V}_4$, and $(\mathbf{c}, \bm{\beta})$ is uniquely identified on $\mathcal{C}_0 \times \mathbb{R}^n$.
\end{enumerate}
\end{proposition}


Geometrically, if $a_k \neq 0$ for some node $k$, matching $d^{*}_{ik}(t)^2 = d_{ik}(t)^2 + a_k$ across all unshifted nodes $i \in \mathcal{U} = \{j : a_j = 0\}$ forces the at least $n - 4$ positions $\{z_i(t)\}_{i \in \mathcal{U}}$ onto a single line in $\mathbb{R}^2$. This collinear degeneracy contradicts Assumption (A3), forcing $\bm{a} = \bm{0}$. 

Beyond global uniqueness, asymptotic theory and model selection require local non-degeneracy. The following Proposition~\ref{prop:fisher} ensures this by establishing strict positive definiteness of the expected Fisher information matrix on $\mathcal{C}_0 \times \mathbb{R}^n$.

\begin{proposition}
\label{prop:fisher}
Let $n \ge 7$ and let $(\mathbf{c}, \bm{\beta}) \in \mathcal{C}_0 \times \mathbb{R}^n$ satisfy (A1), (A2), and (A3). Then the mapping $(\mathbf{c}, \bm{\beta}) \mapsto \{ \log \lambda_{ij}(\cdot) \}_{(i,j) \in \mathcal{D}}$ has an injective differential on the $(2n-3)K + n$ free coordinates, and the expected Fisher information matrix
\begin{equation}
  \mathcal{I}(\mathbf{c}, \bm{\beta}) = \sum_{(i,j) \in \mathcal{D}} \int_0^T \nabla \log \lambda_{ij}(t) \, \nabla \log \lambda_{ij}(t)^{\top} \lambda_{ij}(t) \, dt
\end{equation}
is strictly positive definite.
\end{proposition}

Strict positive definiteness of $\mathcal{I}(\mathbf{c}, \bm{\beta})$, combined with $\bm{\Pi} \succ 0$ in \eqref{eq:Pi}, guarantees local strict concavity of the penalized log-likelihood surface. Consequently, the approximation of BIC in Section~\ref{subsec:bic} is mathematically well-posed.


%% file: tex_files/4_estimation.tex
\section{Estimation}
\label{sec:estimation}

To avoid joint optimization over a non-convex parameter space, we decouple \eqref{eq:objective} into two conditionally efficient updates. Given $\bm{\beta}$, Section~\ref{subsec:sgd} develops a stochastic gradient ascent scheme over subsampled dyad minibatches. Given $\mathbf{c}$, Section~\ref{subsec:beta-update} solves the strictly concave profile problem in $\bm{\beta}$ via a parallel fixed-point update. Both blocks rely on Gauss--Legendre quadrature to evaluate cumulative dyadic intensities (Section~\ref{subsec:quadrature}). To avoid local minima from trajectory misalignment, Section~\ref{subsec:init} leverages the idea of optimal transport alignment to construct a piecewise-linear initializer, while Section~\ref{subsec:bic} selects smoothing parameters via a curvature-adjusted Bayesian information criterion accelerated by randomized trace estimation.Algorithm~\ref{alg:main} summarizes the complete procedure.

\begin{algorithm}[t]
\caption{Estimation of the \dlspp{} model}
\label{alg:main}
\begin{algorithmic}[1]
  \State \textbf{Input:} events $\{\mathcal{H}_{ij}\}$, basis $\{\phi_k\}$,
    quadrature $\{(s_g,w_g)\}$, tuning $(\rho_1,\rho_2,\tau^2)$, batch size $B$,
    step sizes $\{\eta_s\}$.
  \State Precompute $\bm{\Phi}$, $\bm{\Omega}_1$, $\bm{\Omega}_2$.
  \State Initialize $\mathbf{c}^{(0)} \in \mathcal{C}_0$ by Section~\ref{subsec:init}; set $\bm{\beta}^{(0)}$ by \eqref{eq:beta-fp} at $\mathbf{c}^{(0)}$.
  \For{$s = 0,1,2,\dots$ until convergence}
    \State \textbf{(a) Trajectory block.} For one pass over a random partition of
      $\mathcal{D}$ into batches of size $B$: draw $\mathcal{B}$, form
      $\widehat{\nabla}_{\mathbf{c}}\mathcal{J}$ by \eqref{eq:sgd-grad}, and update
      $\mathbf{c} \leftarrow \mathbf{c} + \eta_s \widehat{\nabla}_{\mathbf{c}}\mathcal{J}$,
      holding anchored entries at zero.
    \State \textbf{(b) Activity block.} Iterate \eqref{eq:beta-fp} to tolerance
      $\epsilon_{\bm{\beta}}$ at the current $\mathbf{c}$.
  \EndFor
  \State \textbf{Output:} $\hat{\mathbf{c}}$, $\hat{\bm{\beta}}$.
\end{algorithmic}
\end{algorithm}

\subsection{Stochastic gradient ascent for the trajectories}
\label{subsec:sgd}

Evaluating the exact gradient $\nabla_{\mathbf{c}}\ell$ demands $O(n^2 Q K)$ operations due to summation over all $\binom{n}{2}$ dyads and $Q$ quadrature nodes. At iteration $s$, we sample a minibatch $\mathcal{B}_s \subset \mathcal{D}$ of $B$ dyads uniformly at random without replacement and construct the unbiased estimator
\begin{equation}
\label{eq:sgd-grad}
  \widehat{\nabla}_{\mathbf{c}}\, \mathcal{J}
  \;=\;
  \frac{\binom{n}{2}}{B}
  \sum_{(i,j) \in \mathcal{B}_s}
  \nabla_{\mathbf{c}}\, \ell_{ij}(\mathbf{c},\bm{\beta})
  \;-\;
  \nabla_{\mathbf{c}}\, \mathcal{P}(\mathbf{c}),
\end{equation}
where $\ell_{ij}$ is the log-likelihood contribution of dyad $(i,j)$ and the cheap penalty gradient is evaluated in full. We update coefficients via $\mathbf{c} \leftarrow \mathbf{c} + \eta_s \widehat{\nabla}_{\mathbf{c}}\mathcal{J}$, holding anchored entries at zero to enforce $\mathbf{c} \in \mathcal{C}_0$. Subsampling reduces the computational cost from $O(n^2 Q K)$ to $O(B Q K)$. A decaying step size $\eta_s$ is used until the relative change in $\mathcal{J}$ falls below tolerance.

\subsection{Numerical evaluation of the cumulative intensity}
\label{subsec:quadrature}

We use $\Lambda_{ij} = \int_0^T \lambda_{ij}(t)\,dt$ to denote the cumulative intensity of dyad $(i,j)$.
Under \eqref{eq:dist-spline} the integrand of $\Lambda_{ij}$ is the exponential of a piecewise polynomial of degree $2p$, so the integral has no closed form and is evaluated
by Gauss--Legendre quadrature \citep{stroud1966gaussian}. Let $\{(s_q, w_q)\}_{q=1}^{Q}$ be the nodes and weights of the rule with $Q$ points on $[0,T]$, so that
\begin{equation}
\label{eq:quad}
  \Lambda_{ij}
  \;\approx\;
  \sum_{q=1}^{Q} w_q \, \exp\bigl\{ \beta_i + \beta_j - d_{ij}(s_q)^2 \bigr\}.
\end{equation}
Although the log-likelihood is re-evaluated at every iteration, the grid is fixed throughout so the basis matrix $\bm{\Phi} \in \mathbb{R}^{Q \times K}$ with $\Phi_{qk} = \phi_k(s_q)$ is formed only once and \eqref{eq:quad} is therefore a weighted sum over a fixed grid for every dyad at every iteration.
Similarly, the Gram matrices \eqref{eq:f1} and \eqref{eq:f2} depend only on the basis and are computed once. 
Further details are given in the Supplement Section 1. 

\subsection{Profile update for the activity parameters}
\label{subsec:beta-update}

Fix $\mathbf{c}$ and recall $I_{ij} = \int_0^T \exp\{-d_{ij}(t)^2\}\,dt$, evaluated by \eqref{eq:quad}. Differentiating \eqref{eq:loglik} with respect to $\beta_i$ gives
\[
  \frac{\partial \ell}{\partial \beta_i}
  \;=\;
  m_i \;-\; e^{\beta_i} \sum_{j \ne i} e^{\beta_j} I_{ij},
\] since node $i$ appears in exactly the $n-1$ dyads containing it and each dyad is counted once. Setting the score to zero yields the system $  e^{\beta_i} \sum_{j \ne i} e^{\beta_j} I_{ij} =m_i$ for $i = 1,\dots,n$,
which has no closed-form solution. We solve it by the parallel fixed-point iteration
\begin{equation}
\label{eq:beta-fp}
  \beta_i^{(s+1)}
  \;=\;
  \log m_i
  \;-\;
  \log \sum_{j \ne i} \exp\bigl( \beta_j^{(s)} + \log I_{ij} \bigr),
  \qquad i = 1,\dots,n,
\end{equation}
evaluated in the log-sum-exp form shown to avoid overflow when $\beta_j$ or $\log I_{ij}$ is large in magnitude. By Proposition~\ref{prop:beta} the limit is the unique maximizer whenever $m_i > 0$ for all $i$. 
The updates for $\beta$ under the homogeneous baseline \eqref{eq:homog} and the Poisson $\beta$-model~\eqref{eq:activity} are included in Supplementary Section 3.


\subsection{Initialization}
\label{subsec:init}

The log-likelihood is nonconvex in the spline coefficients, so the optimizer may be trapped in local minima, and initialization matters in practice. We therefore construct a starting configuration that carries the dominant latent movement, using a two-stage piecewise-linear procedure.

\paragraph{Stage 1: block-wise static fits}
We first partition $[0,T]$ into subintervals $[0,T/2]$ and $[T/2,T]$. Within each block $b \in \{1,2\}$, we fit a static Poisson latent space model with cumulative intensity
\begin{equation}
\label{eq:static-fit}
  \log \Lambda^{(b)}_{ij}
  \;=\;
  \log(T/2) \;+\; \beta^{(b)}_i + \beta^{(b)}_j \;-\; \norm{ w^{(b)}_i - w^{(b)}_j }^2.
\end{equation}
Positions $\mathbf{W}^{(b)} \in \mathbb{R}^{n \times 2}$ and node intercepts are estimated jointly by maximizing the block log-likelihood with a light ridge penalty $\frac{\gamma}{2}\sum_i \norm{w^{(b)}_i}^2$. 
We discard the fitted intercepts, retaining only $\mathbf{W}^{(1)}$ and $\mathbf{W}^{(2)}$.
Each block is anchored as in \eqref{eq:C0}, so each is determined only up to $\mathbb{V}_4$.

\paragraph{Stage 2: alignment and projection}
Since each static block is anchored independently, directly connecting $\mathbf{W}^{(1)}$ and $\mathbf{W}^{(2)}$ across time can artificially twist or invert node trajectories. Proposition~\ref{prop:ident} establishes that this mismatch is restricted to coordinate axis sign flips in $\mathbb{V}_4$. We test $\mathbf{W}^{(2)}$ against all four candidates in $\mathbb{V}_4$ to minimize total spatial displacement
\begin{equation}
\label{eq:align}
  \mathbf{U}^{*}
  \;=\;
  \operatorname{arg\,min}_{\mathbf{U} \in \mathbb{V}_4}
  \sum_{i=1}^{n} \norm{ w^{(1)}_i - \mathbf{U}^{\top} w^{(2)}_i }^2,
  \qquad
  \mathbf{W}^{(2)*} = \mathbf{W}^{(2)}\mathbf{U}^{*}.
\end{equation}
This discrete Procrustes alignment frames gauge selection as an optimal transport problem with minimal Euclidean displacement cost. We then connect $w^{(1)}_i$ and $w^{(2)*}_i$ linearly and project these paths onto the B-spline basis via least squares yields $\mathbf{c}^{(0)} \in \mathcal{C}_0$. This warm start captures broad movement trends while preventing erratic path oscillations during early optimization steps. Figure~\ref{fig:initialization} illustrates the construction on a simulated network.


\begin{figure}[t]
  \centering
  \begin{subfigure}[t]{0.32\linewidth}
    \centering
    \includegraphics[width=\linewidth]{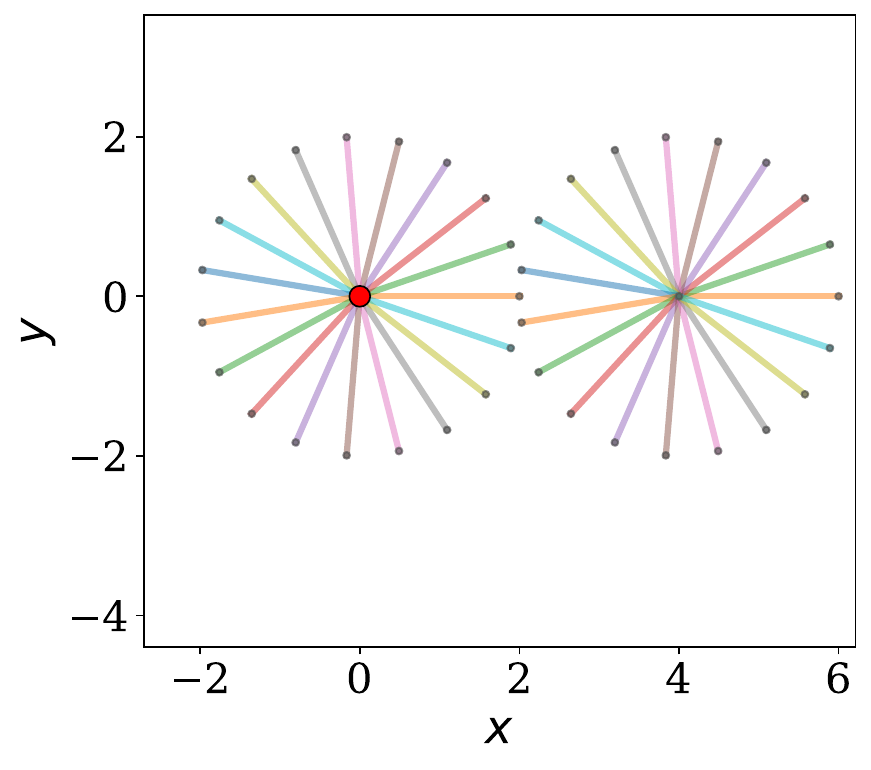}
    \caption{True trajectories.}
    \label{fig:init-true}
  \end{subfigure}\hfill
  \begin{subfigure}[t]{0.32\linewidth}
    \centering
    \includegraphics[width=\linewidth]{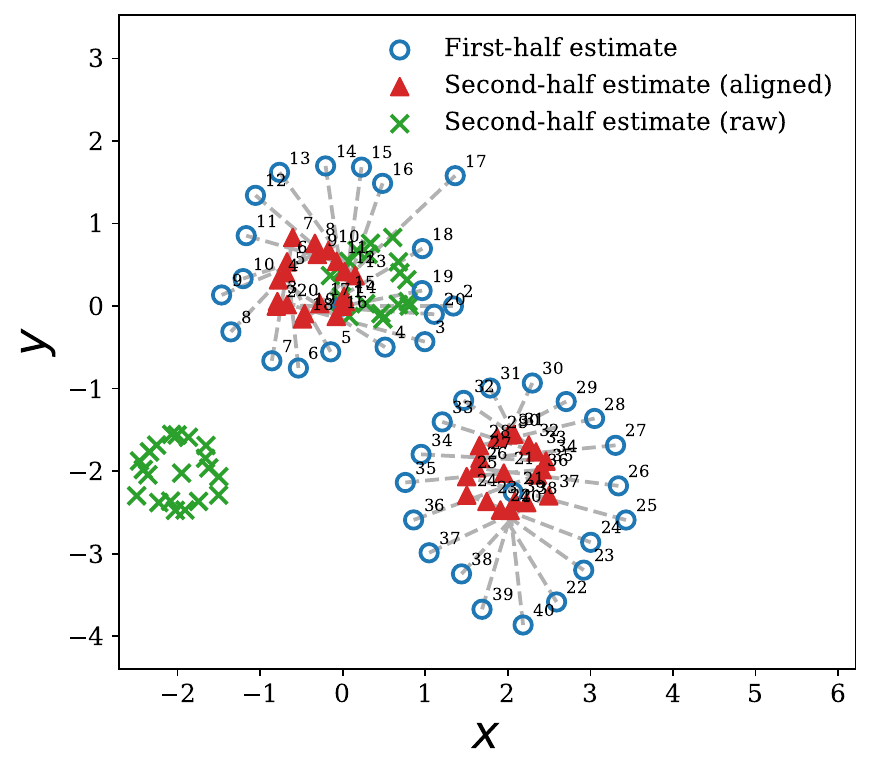}
    \caption{Block-wise static fits.}
    \label{fig:init-static}
  \end{subfigure}\hfill
  \begin{subfigure}[t]{0.32\linewidth}
    \centering
    \includegraphics[width=\linewidth]{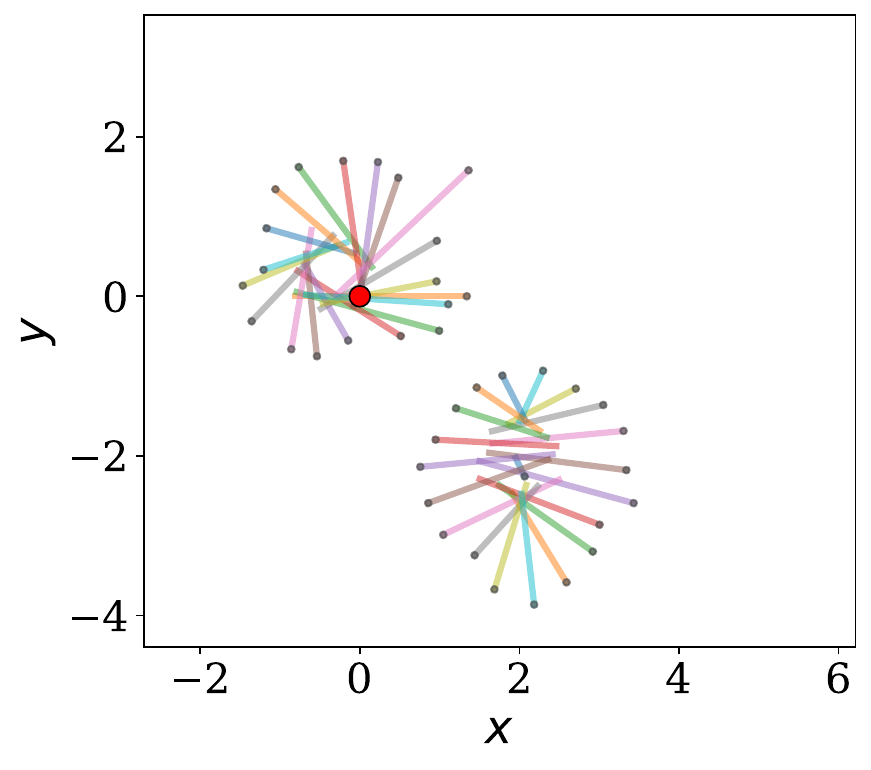}
    \caption{Piecewise-linear initializer.}
    \label{fig:init-linear}
  \end{subfigure}
  \caption{Initialization of the spline coefficients. Static models in (b) are fit on $[0,T/2]$ and $[T/2,T]$ separately; each is anchored as in \eqref{eq:C0} and is therefore determined only up to the four elements of $\mathbb{V}_4$. The orientation minimizing \eqref{eq:align} is retained, and connecting the two configurations in time gives (c), which is projected onto the B-spline basis to obtain $\mathbf{c}^{(0)}$.}
  \label{fig:initialization}
\end{figure}

\subsection{Selection of tuning parameters}
\label{subsec:bic}

We select smoothing parameters $(\rho_1,\rho_2)$ over a grid by minimizing the Bayesian information criterion \citep{schwarz1978estimating}
\begin{equation}
\label{eq:bic}
  \mathrm{BIC}(\rho_1,\rho_2)
  \;=\;
  -2\,\ell(\hat{\mathbf{c}}, \hat{\bm{\beta}})
  \;+\;
  \bigl\{ k_{\mathrm{eff}}(\rho_1,\rho_2) + n \bigr\} \log m,
\end{equation}
where $m$ is the total event count across all dyads and $k_{\mathrm{eff}}$ is the effective degrees of freedom of the trajectory component \citep{gu2005generalized}. Standard trace formulas rely on the log-likelihood Hessian, but non-convex distance mappings produce negative eigenvalues away from local modes. We therefore substitute the positive-semidefinite Gauss--Newton information matrix $\mathbf{F}$ for the Hessian, which discards second-derivative distance terms while capturing pure data curvature, yielding
\begin{equation}
\label{eq:keff}
  k_{\mathrm{eff}}
  \;=\;
  \tr\bigl\{ \mathbf{F}\,(\mathbf{F} + \mathbf{H}_{\mathcal{P}})^{-1} \bigr\},
\end{equation}
where $\mathbf{H}_{\mathcal{P}} = \mathbf{I}_{2n-3} \otimes \bm{\Pi}$ represents the penalization block. To scale to large networks, we evaluate \eqref{eq:keff} via a Hutchinson randomized trace estimator \citep{hutchinson1989stochastic} accelerated by preconditioned conjugate gradients. Supplementary~Section 4 details the curvature construction, linear system implementations, and optimization surface dynamics.


%% file: tex_files/5_simulation.tex
\section{Simulation studies}
\label{sec:simulation}

We conduct two simulation studies to evaluate the in-sample and out-of-sample empirical performance of the proposed dynamic latent space model. Motivated by international relations networks, our designs consider both fluid multipolar movements and structured bipolar coalition formations:
\begin{itemize}
    \item \textbf{S1: Fluid multipolar drift:} $N$ nodes start from random initial locations and move along smooth, continuous trajectories generated via a Fourier basis expansion. This setup models unconstrained diplomatic wandering without imposing global spatial symmetries.
    \item \textbf{S2: Bipolar alliance consolidation:} $N$ nodes are partitioned into two dynamic clusters centered at $(0,0)$ and $(4,0)$, contracting inward toward their centers over time. This setup tests the model's ability to capture sharp clustering and converging pairwise distances.
\end{itemize}
Across both settings, true activity parameters $\bm{\beta}$ are linearly spaced from $1.5$ to $2.5$ across the $N$ nodes. With an average of 4.8 interactions per dyad in S1 and 27.1 in S2, these settings evaluate the model across contrasting sparse and dense network regimes.
Figure~\ref{fig:simulations_true_initial_estimate} displays the true, initial and estimated trajectories of each design.

\begin{figure}[h]
\centering
\setlength{\tabcolsep}{0pt}
\begin{tabular}{@{}c m{0.31\textwidth} m{0.31\textwidth} m{0.31\textwidth}@{}}
 & \centering\textbf{True} & \centering\textbf{Initialization} & \centering\arraybackslash\textbf{Estimation} \\[0.15cm]
\textbf{S1} &
\includegraphics[width=\linewidth]{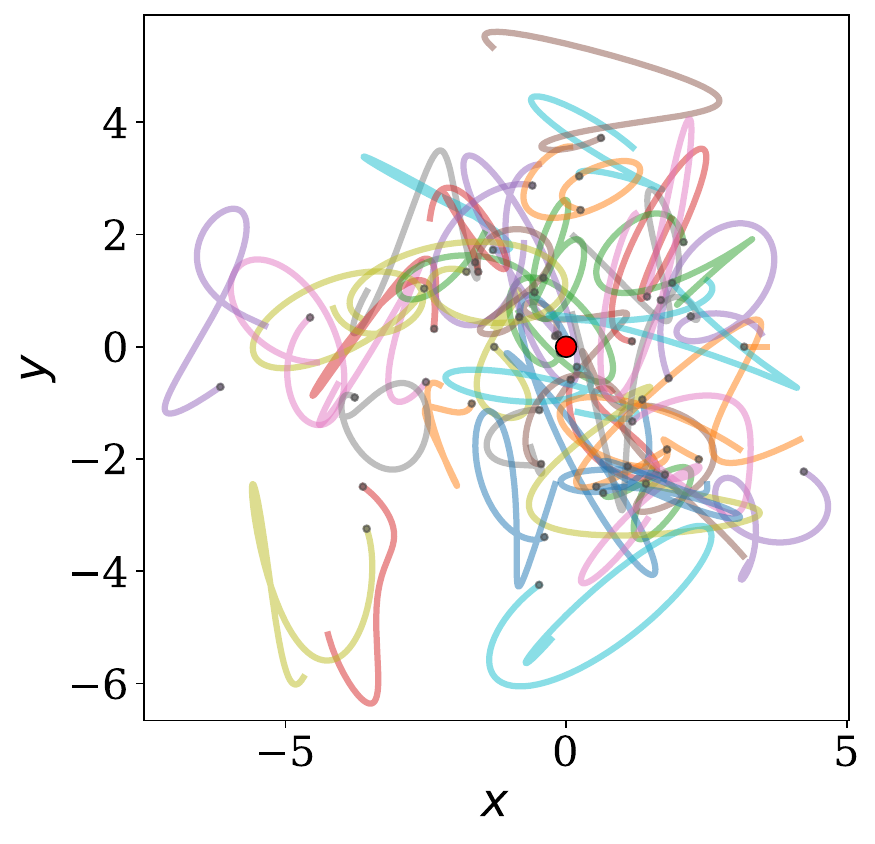} &
\includegraphics[width=\linewidth]{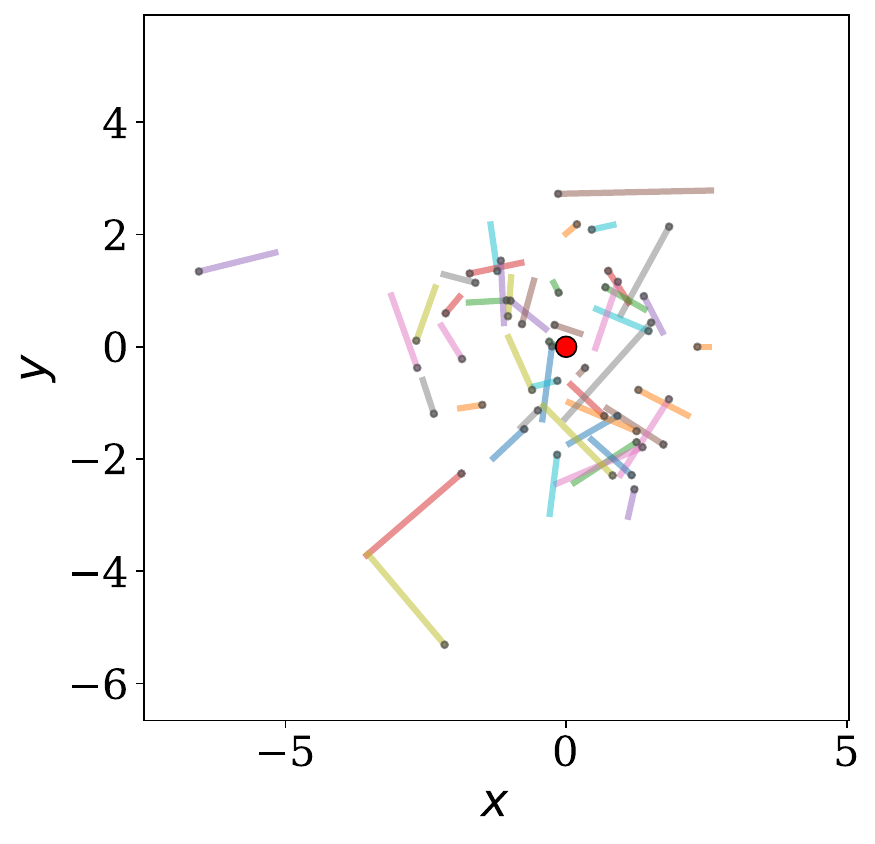} &
\includegraphics[width=\linewidth]{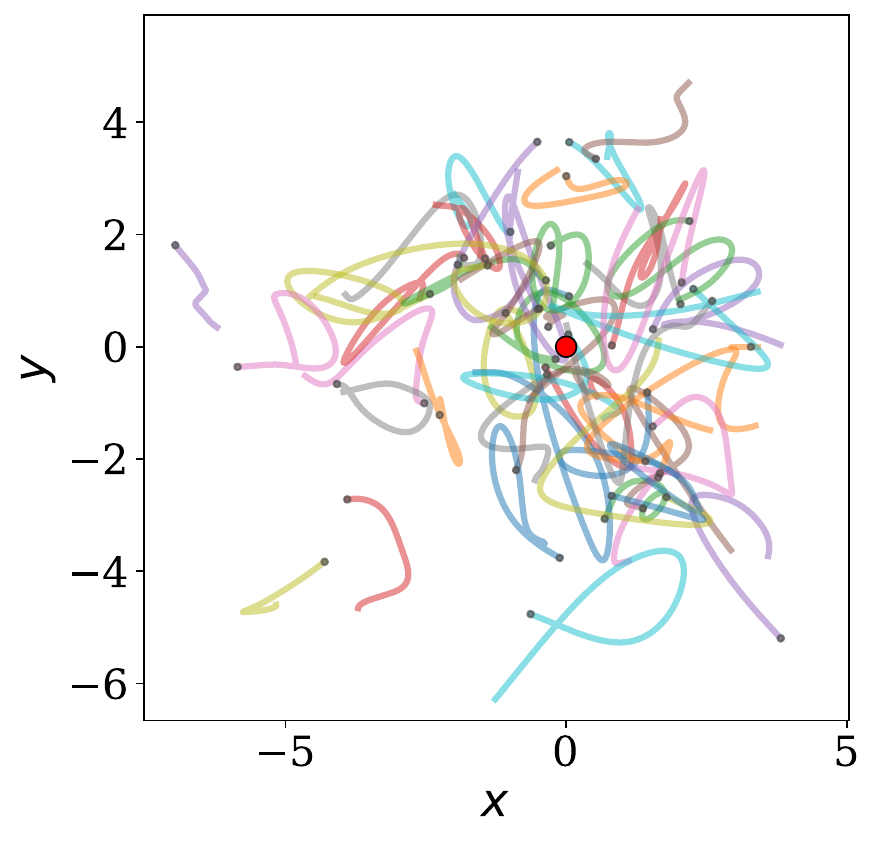} \\
\textbf{S2} &
\includegraphics[width=\linewidth]{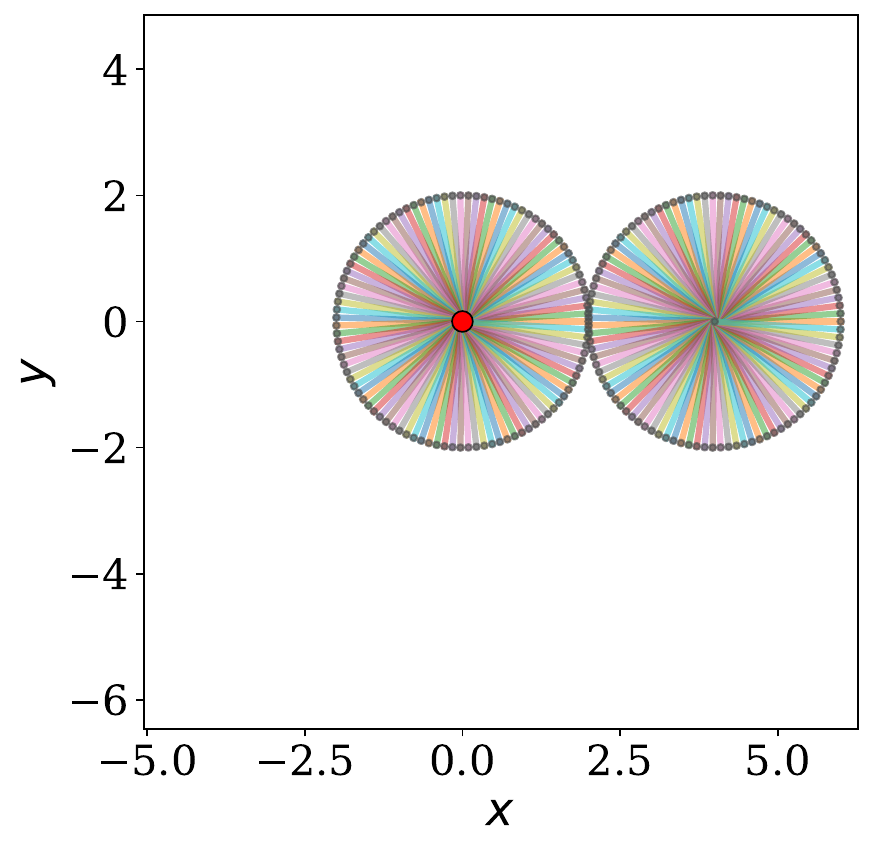} &
\includegraphics[width=\linewidth]{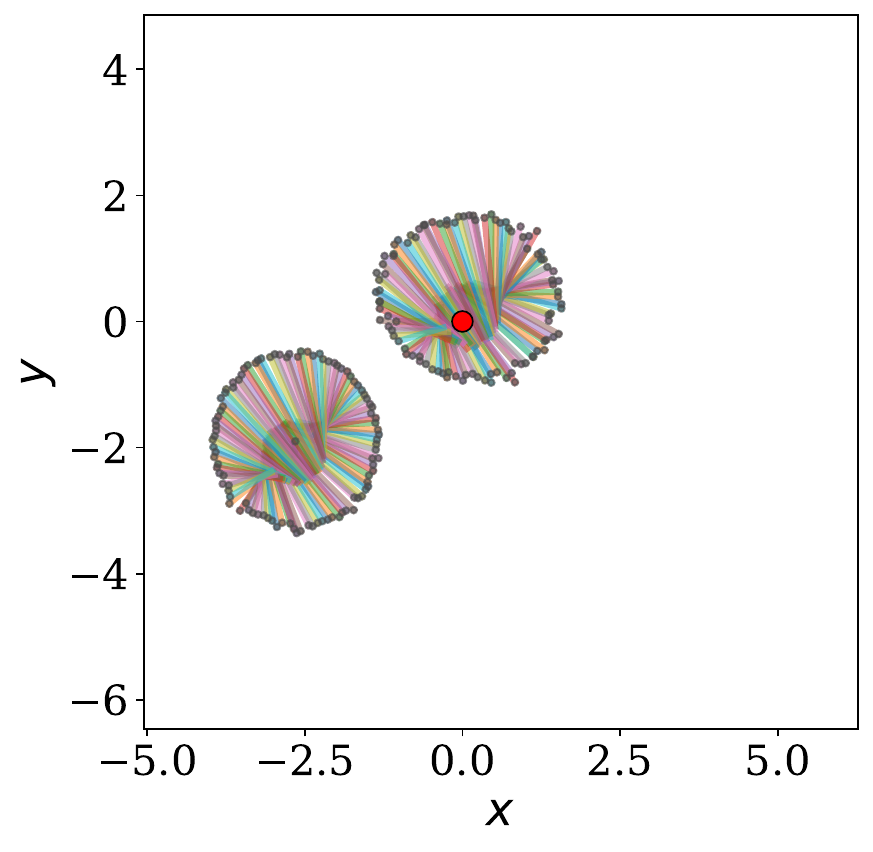} &
\includegraphics[width=\linewidth]{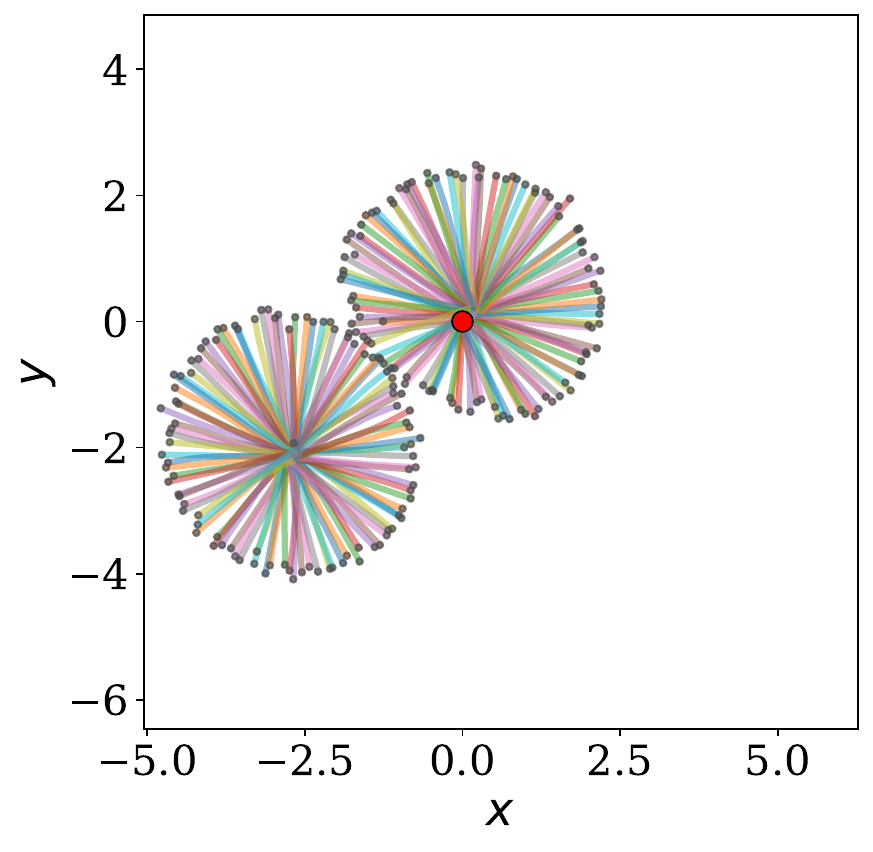}
\end{tabular}
\caption{True (left), initial (middle) and estimated (right) latent trajectories for one replication of designs S1 (top) and S2 (bottom). Each curve is one node's path over $[0,T]$. Grey points mark $t=0$ and the red point marks the origin.}
\label{fig:simulations_true_initial_estimate}
\end{figure}

\paragraph*{\textbf{Comparison baselines}}
We compare the proposed DLS-PP against four baselines. The first two methods were introduced in Section~\ref{sec:model}, the homogeneous baseline~\eqref{eq:homog} where $\log \lambda_{ij} (t) = \beta -\|z_i (t) - z_j (t) \|^2$ with $z_i(t)$ modeled by B-spline and Poisson $\beta$-model~\eqref{eq:activity} where $\log \lambda_{ij} (t) = \beta_i + \beta_j$. As two more comparisons, we consider:
\begin{itemize}
    \item The \textbf{continuous latent position models of \cite{rastelli_2023_continuous} (CLPM):} This is the closest method to ours where $\log \lambda_{ij} (t)=\beta - \|z_i (t) - z_j (t) \|^2$ and the trajectories $z_i (t)$ is assumed to be piece-wise linear curves. 
    \item A \textbf{dyad-independent Poisson process (DIPP):} This simple baseline estimates each pairwise conditional intensity function $\log \lambda_{ij}(t)$ independently by fitting a B-spline. 





\end{itemize}

\paragraph*{\textbf{In-sample and out-of-sample}}
We evaluate the proposed method from two complementary perspectives: in-sample estimation and out-of-sample prediction. 
For in-sample evaluation, the model is fit using the full observation period, with tuning parameters selected according to the proposed BIC. For out-of-sample evaluation, we divide the observation period into a training interval $[0,T_{\mathrm{tr}}]$ and a held-out interval $(T_{\mathrm{tr}},T]$. In this section, $T_{\mathrm{tr}}$ is set as $0.9$ in all experiments. The model is fitted and tuned using the training data, with the estimated $\beta$ and B-spline coefficients fixed for prediction on the held-out interval. Since the latent trajectories represented by B-splines are not directly observed beyond $T_{\mathrm{tr}}$, we consider two forecasting strategies for extending the estimated trajectories into the held-out interval. Under \emph{persistence (P)}, each node is held fixed at its
estimated terminal position in the training interval,
\[
\hat{\mathbf z}_i^{\,\mathrm{P}}(t)
=
\hat{\mathbf z}_i(T_{\mathrm{tr}}),
\qquad
t\in(T_{\mathrm{tr}},T],
\]
whereas under \emph{constant-velocity extrapolation (CV)}, each trajectory is extended linearly using its estimated terminal velocity,
\[
\hat{\mathbf z}_i^{\,\mathrm{CV}}(t)
=
\hat{\mathbf z}_i(T_{\mathrm{tr}})
+
(t-T_{\mathrm{tr}})
\hat{\mathbf z}_i'(T_{\mathrm{tr}}),
\qquad
t\in(T_{\mathrm{tr}},T].
\]

\paragraph*{\textbf{Performance evaluation}}
We use in-sample data to assess the recovery of the latent trajectories and model parameters, whereas the out-of-sample evaluates the ability of the fitted model to predict future latent configurations and observed events. 
In both in-sample and out-of-sample data, we assess trajectory recovery using the integrated mean squared error of the pairwise distances, $\mathrm{IMSE}_d$:
\[
\mathrm{IMSE}_d = \frac{1}{T \binom{N}{2}} \int_0^T \sum_{1 \le i < j \le N} \Bigl( \hat{d}_{ij}(t) - d_{ij}(t) \Bigr)^2 \, dt, \qquad \mathrm{RMSE}_d = \sqrt{\mathrm{IMSE}_d}.
\]
In practice, we evaluate the integral via a Riemann sum over a temporal grid $0 = t_1 < \dots < t_K = T$ with step sizes $\Delta t_k = t_{k+1} - t_k$:
\[
\widehat{\mathrm{IMSE}}_d = \frac{1}{T \binom{N}{2}} \sum_{k=1}^{K-1} \left[ \sum_{1 \le i < j \le N} \Bigl( \hat{d}_{ij}(t_k) - d_{ij}(t_k) \Bigr)^2 \right] \Delta t_k, \qquad \widehat{\mathrm{RMSE}}_d = \sqrt{\widehat{\mathrm{IMSE}}_d}.
\]

We additionally report the Kolmogorov--Smirnov (KS) statistic based on the time-rescaling theorem \citep{brown2002time} to assess the goodness of fit of the fitted point process in both in-sample and out-of-sample data. 

With in-sample data, representative estimated and true latent trajectories are also visualized to illustrate the quality of trajectory recovery. Moeover, parameter recovery using the estimation error of the node-specific effects ${\beta_i}$. 
With held-out data, we also assess event-count prediction. For each dyad $(i,j)$, the predicted expected number of events in the held-out interval is $\hat\mu_{ij}:=\int_{T_{\mathrm{tr}}}^{T} \hat\lambda_{ij}(t)\,dt$ so that the model implies $ N_{ij}^{\mathrm{out}} \sim \operatorname{Poisson}\!\left(\hat\mu_{ij} \right).$
We compare these predicted expected counts with the held-out counts $N_{ij}^{\mathrm{out}}$. All out-of-sample metrics are computed separately for the persistence and constant-velocity forecasting strategies. More details are given in Supplementary Section 5.

\subsection{S1: Fluid multipolar drift}

The first simulation evaluates model performance in an unconstrained, non-linear setting representing fluid multipolar state interactions. We generate trajectories for $N=50$ nodes in $\mathbb{R}^2$ over $t \in [0,1]$. Initial positions are drawn from a standard normal distribution around the origin, $\mathbf{z}_i(0) \sim \mathcal{N}(\mathbf{0}, \sigma_0^2 \mathbf{I}_2)$, with nodes $1$ and $2$ adhering to the identification constraints defined above. For all remaining nodes $i \in \{3, \dots, N\}$, trajectories are constructed via a smooth Fourier basis expansion around their initial positions:
\[
\mathbf{z}_i(t) = \mathbf{z}_i(0) + \sum_{m=1}^{M} \Bigl[ \mathbf{a}_{i,m} \sin(m \pi t) + \mathbf{b}_{i,m} \bigl(1 - \cos(m \pi t)\bigr) \Bigr], \qquad t \in [0,1],
\]
where $\mathbf{a}_{i,m}, \mathbf{b}_{i,m} \sim \mathcal{N}(\mathbf{0}, \sigma_v^2 \mathbf{I}_2)$ are random velocity vectors and $M=3$ controls path curvature. This setup produces random dynamics without spatial symmetries, rigorously testing the flexibility of our trajectory regularizers. The true trajectories are shown in Figure~\ref{fig:simulations_true_initial_estimate}.

\begin{figure}[htbp]
    \centering
    \begin{subfigure}[t]{.32\textwidth}
        \centering
        \includegraphics[width=\linewidth]{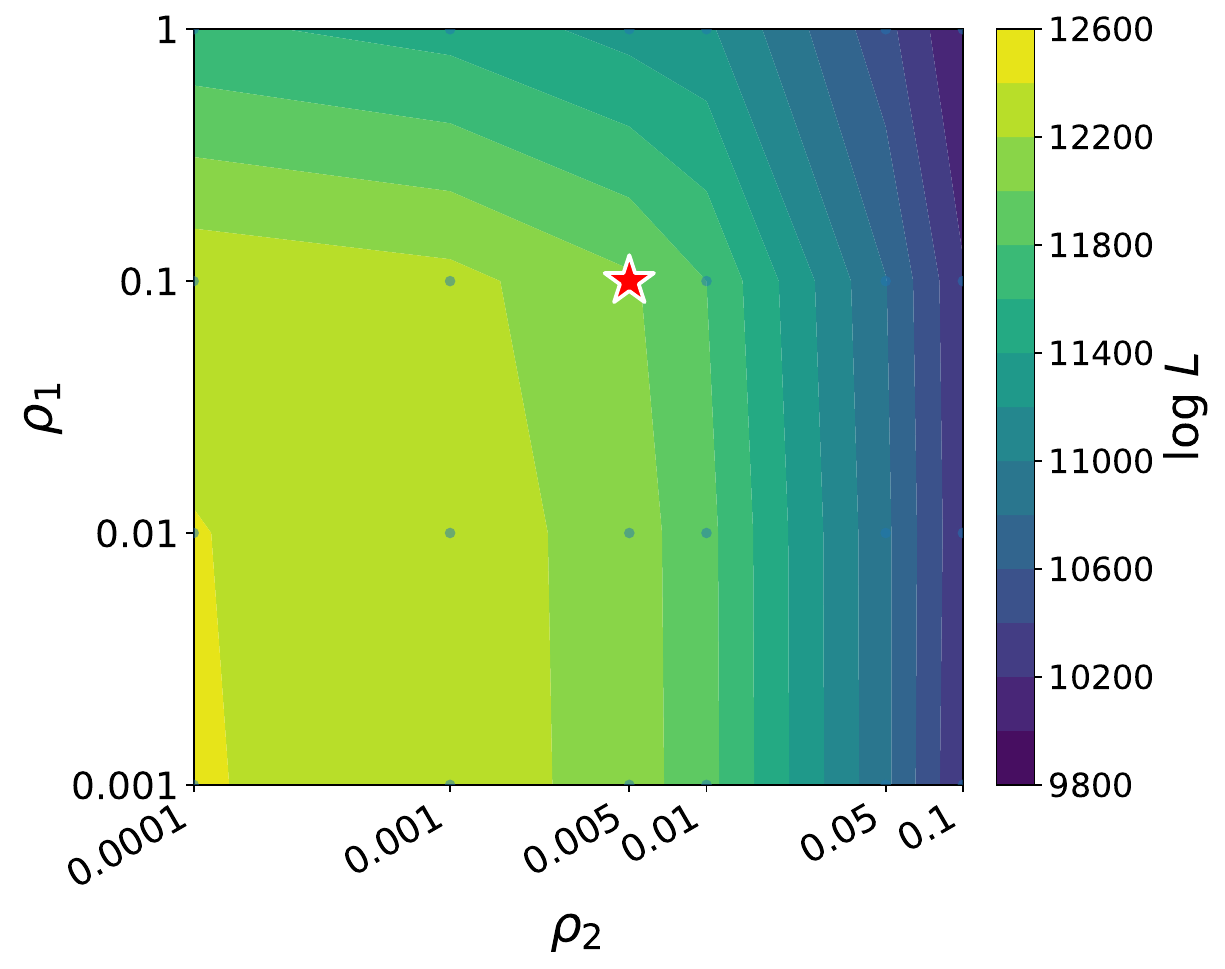}
        \subcaption{Log-likelihood}
        \label{fig:fluid_drift_loglik_surface}
    \end{subfigure}\hfill
    \begin{subfigure}[t]{.32\textwidth}
        \centering
        \includegraphics[width=\linewidth]{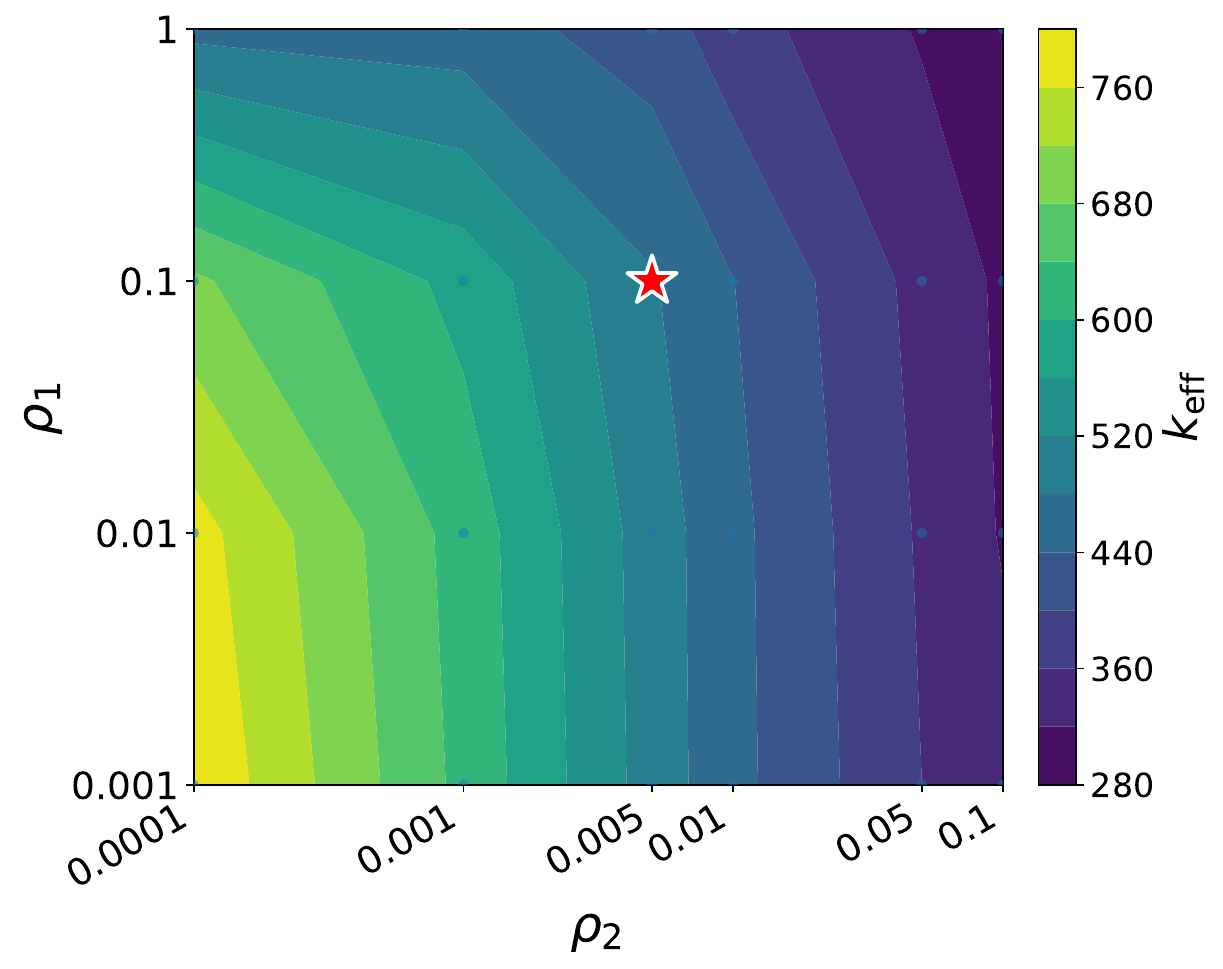}
        \subcaption{Effective degrees of freedom}
\label{fig:fluid_drift_keff_surface}
    \end{subfigure}\hfill
    \begin{subfigure}[t]{.32\textwidth}
        \centering
        \includegraphics[width=\linewidth]{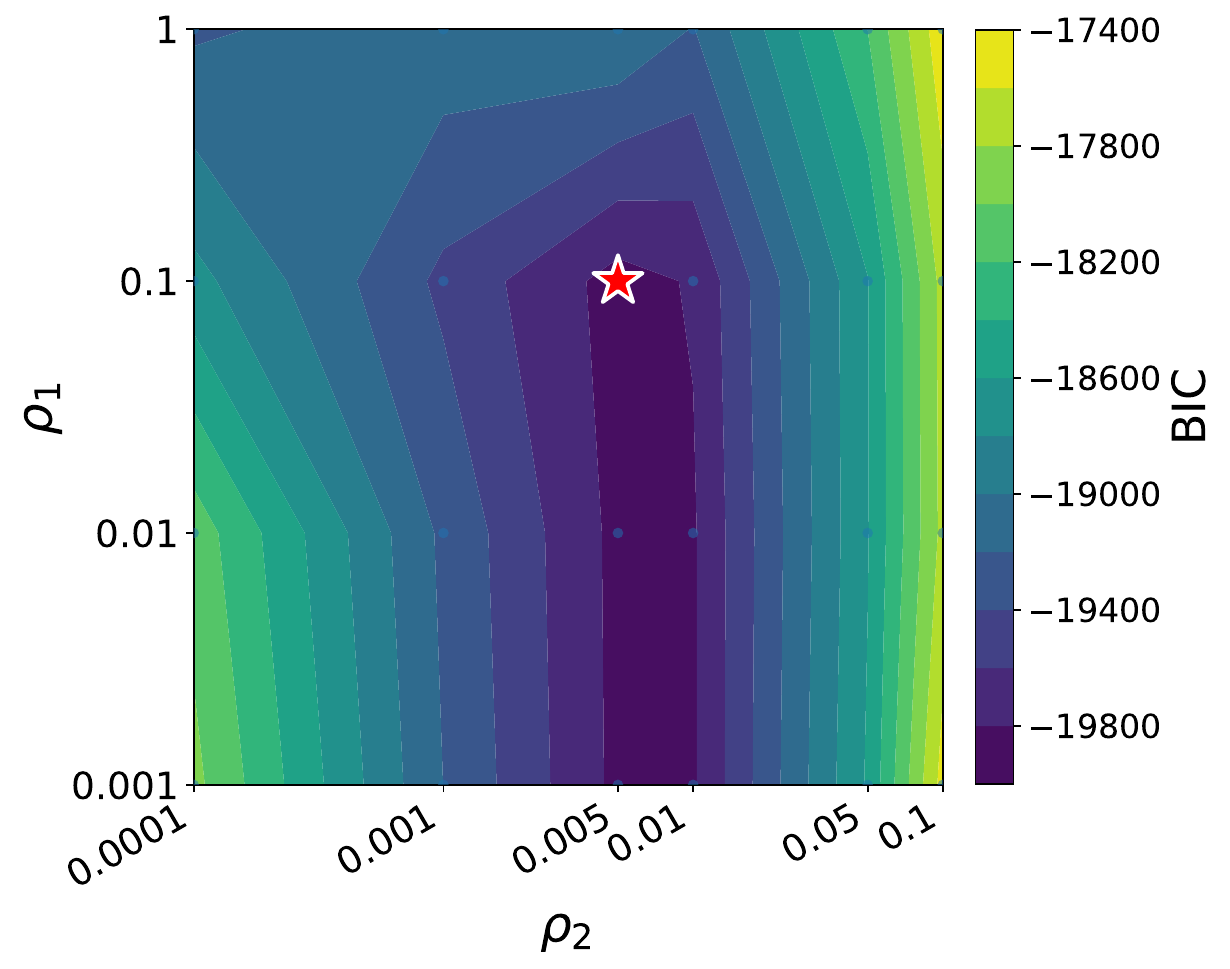}
        \subcaption{BIC}
        \label{fig:fluid_drift_bic_surface}
    \end{subfigure}
    \caption{Tuning surfaces over the two-dimensional $(\rho_1,\rho_2)$ grid in S1. The star marks the selected pair, minimizing BIC.}
    \label{fig:fluid_drift_BIC}
\end{figure}

Figure~\ref{fig:fluid_drift_BIC} reports the log-likelihood, effective degrees of freedom $k_{\text{eff}}$, and BIC evaluated over a two-parameter tuning grid of $(\rho_1,\rho_2)$. Both the log-likelihood and $k_{\text{eff}}$ decrease as either penalty increases, as expected, since stronger penalization trades fit for smoothness. BIC balances the two and attains an interior minimum, as indicated by the star in the figure. All three surfaces vary far less along $\rho_1$ than along $\rho_2$, indicating that fit and selection are driven primarily by $\rho_2$ over the ranges tested.

Figure~\ref{fig:S1} shows the estimation of $\beta$ over 30 replications and the true and estimated distance of 18 randomly selected pairs. Table~\ref{tab:S1} reports the mean and standard error over 30 trials, DLS-PP achieves the good predictive performance across most metrics.

\begin{figure}[htbp]
\centering
\begin{subfigure}[t]{0.31\textwidth}
    \centering
    \includegraphics[height=4.5cm]{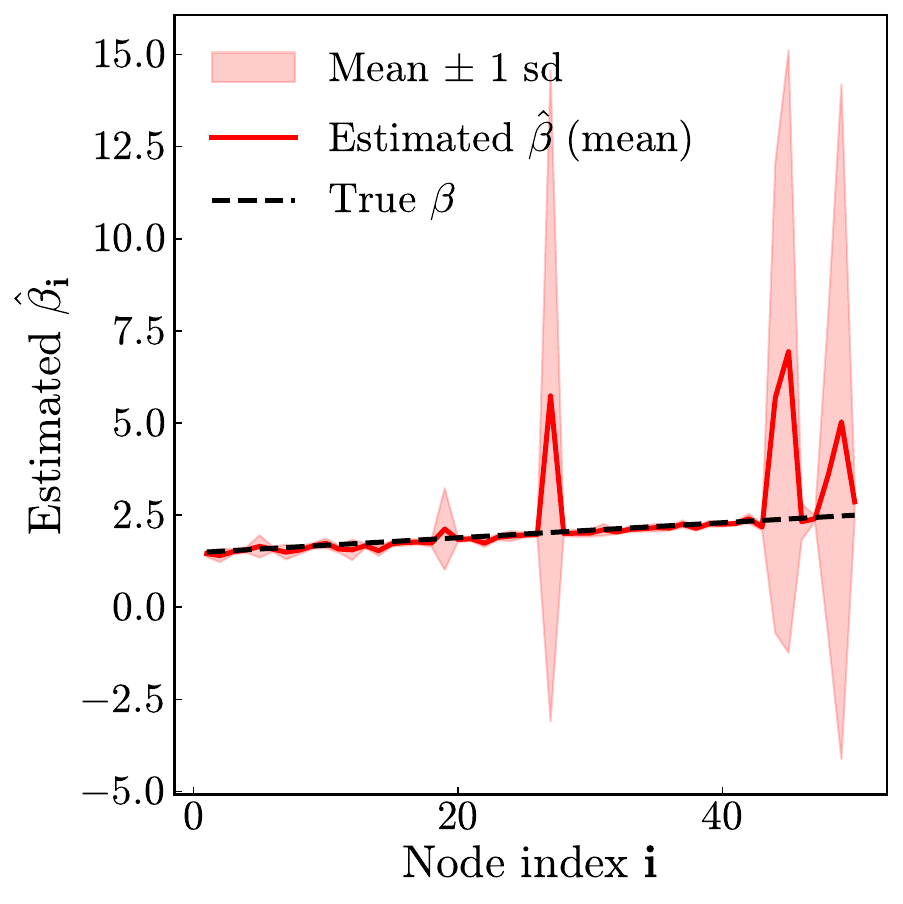}
    \caption{$\beta$ estimation}
    \label{fig:left}
\end{subfigure}\hfill
\begin{subfigure}[t]{0.63\textwidth}
    \centering
    \includegraphics[height=5cm]{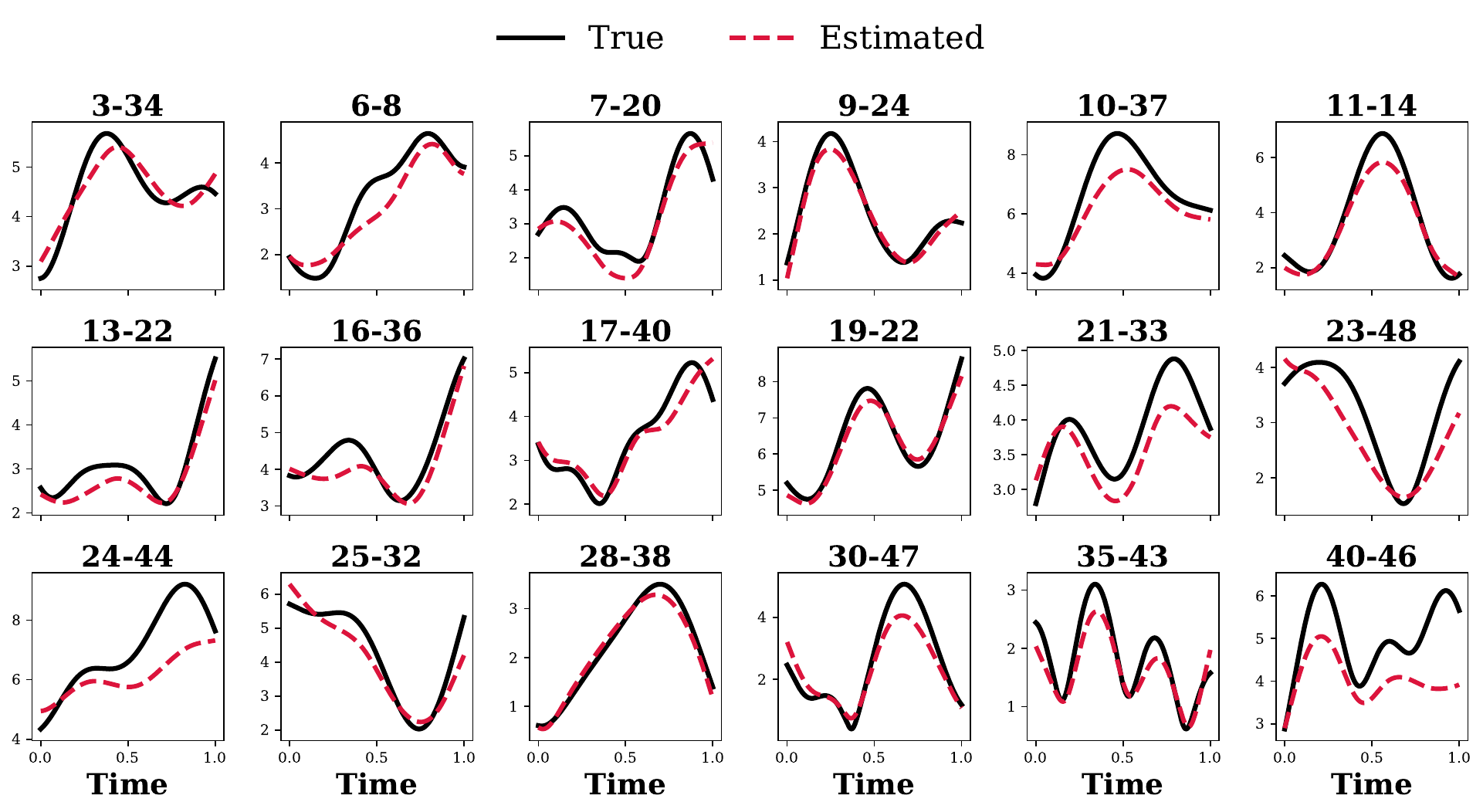}
    \caption{Randomly selected pairwise distances}
    \label{fig:right}
\end{subfigure}
\caption{Estimation of $\beta$ across 30 trials and of the latent distance in one trial under S1. Both components are recovered well, with a few outlying $\beta_i$ estimates expected under this extremely sparse regime.}
\label{fig:S1}
\end{figure}

\begin{table}[t]
\centering
\caption{(S1) Comparison of methods on in-sample fit and two out-of-sample scenarios.}
\label{tab:S1}
\footnotesize
\renewcommand{\arraystretch}{1.15}
\begin{tabular}{llccccc}
\toprule
Scenario & Metric & CLPM & DIPP & \makecell[c]{Poisson-$\beta$\\model} & \makecell[c]{Homogeneous\\baseline} & DLS-PP \\
\midrule
\multirow{3}{*}{In-sample}
  & IMSE            & \makecell[c]{1.5715 \\ (0.3967)} & N/A & N/A & \makecell[c]{0.8389 \\ (0.1607)} & \makecell[c]{\textbf{0.5383} \\ (0.4080)} \\
\cmidrule(l){2-7}
  & KS              & \makecell[c]{0.1069 \\ (0.0382)} & \makecell[c]{0.0553 \\ (0.0039)} & \makecell[c]{0.4520 \\ (0.0053)} & \makecell[c]{0.0484 \\ (0.0040)} & \makecell[c]{\textbf{0.0224} \\ (0.0130)} \\
  \cmidrule(l){2-7}
  & Err $\beta$     & \makecell[c]{1.5646 \\ (0.1328)} & N/A & \makecell[c]{1.9081 \\ (0.0204)} & \makecell[c]{1.7739 \\ (0.0184)} & \makecell[c]{\textbf{1.1264} \\ (1.1404)} \\
\midrule
\multirow{4}{*}{\makecell[l]{Out-of-sample\\Persistence}}
  & IMSE            & \makecell[c]{4.6412 \\ (0.7737)} & N/A & N/A & \makecell[c]{\textbf{1.6882} \\ (0.6525)} & \makecell[c]{1.7648 \\ (0.9096)} \\
\cmidrule(l){2-7}
  & KS              & \makecell[c]{0.3449 \\ (0.0258)} & \makecell[c]{0.2670 \\ (0.0294)} & \makecell[c]{0.4648 \\ (0.0155)} & \makecell[c]{0.2468 \\ (0.0148)} & \makecell[c]{\textbf{0.2418} \\ (0.0179)} \\
\cmidrule(l){2-7}
  & \makecell[l]{Log-\\likelihood} & \makecell[c]{-712.2064 \\ (507.4203)} & \makecell[c]{-1356.5623 \\ (360.5870)} & \makecell[c]{-109.3386 \\ (27.7944)} & \makecell[c]{154.6839 \\ (600.1946)} & \makecell[c]{\textbf{297.4742} \\ (106.0829)} \\
\cmidrule(l){2-7}
  & Count           & \makecell[c]{\textbf{0.7206} \\ (0.0446)} & \makecell[c]{2.0605 \\ (0.5539)} & \makecell[c]{0.8498 \\ (0.0380)} & \makecell[c]{0.7516 \\ (0.0410)} & \makecell[c]{0.7774 \\ (0.0660)} \\
\midrule
\multirow{4}{*}{\makecell[l]{Out-of-sample\\Constant-velocity}}
  & IMSE            & \makecell[c]{10.4659 \\ (1.8391)} & N/A & N/A & \makecell[c]{2.1064 \\ (0.8600)} & \makecell[c]{\textbf{2.0988} \\ (0.9905)} \\
\cmidrule(l){2-7}
  & KS              & \makecell[c]{0.4520 \\ (0.0460)} & \makecell[c]{0.2891 \\ (0.0247)} & N/A & \makecell[c]{0.2769 \\ (0.0186)} & \makecell[c]{\textbf{0.2678} \\ (0.0213)} \\
\cmidrule(l){2-7}
  & \makecell[l]{Log-\\likelihood} & \makecell[c]{-3888.1787 \\ (1412.4392)} & \makecell[c]{$-7.019 \times 10^{13}$ \\ ($7.5 \times 10^{12}$)} & N/A & \makecell[c]{-76.3435 \\ (703.7006)} & \makecell[c]{\textbf{79.6065} \\ (170.5268)} \\
\cmidrule(l){2-7}
  & Count           & \makecell[c]{0.8051 \\ (0.0482)} & \makecell[c]{$1.296 \times 10^{11}$ \\ ($7.31 \times 10^{9}$)} & N/A & \makecell[c]{\textbf{0.7697} \\ (0.0459)} & \makecell[c]{1.0001 \\ (0.4870)} \\
\bottomrule
\end{tabular}
\end{table}




\subsection{S2: Bipolar alliance consolidation}

The second simulation examines a structured configuration modeling bipolar alliance formation and bloc consolidation. We consider $N=200$ nodes in $\mathbb{R}^2$ over $t \in [0,1]$, partitioned equally into two communities centered at $\mathbf{c}_A = (0,0)$ and $\mathbf{c}_B = (4,0)$. Node $1$ is anchored at $\mathbf{c}_A$ and node $101$ is anchored at $\mathbf{c}_B$ across all time. The remaining nodes begin on circles of radius $r=2$ centered at their respective group origin, with initial positions $\mathbf{z}_i(0)$ equally spaced in angle. Each node moves linearly toward its group center at constant velocity according to:
\[
\mathbf{z}_i(t) = (1-t) \, \mathbf{z}_i(0) + t \, \mathbf{c}_g, \qquad t \in [0,1],
\]
for group $g \in \{A, B\}$. Each cluster forms a contracting ring that collapses deterministically into its focal pole, as shown in Figure~\ref{fig:simulations_true_initial_estimate}. This geometric structure tests the model's ability to recover sharp dynamic clustering and resolve pairwise distances during spatial collapse.

\begin{figure}[htbp]
\centering
\begin{subfigure}[t]{0.31\textwidth}
    \centering
    \includegraphics[height=4.5cm]{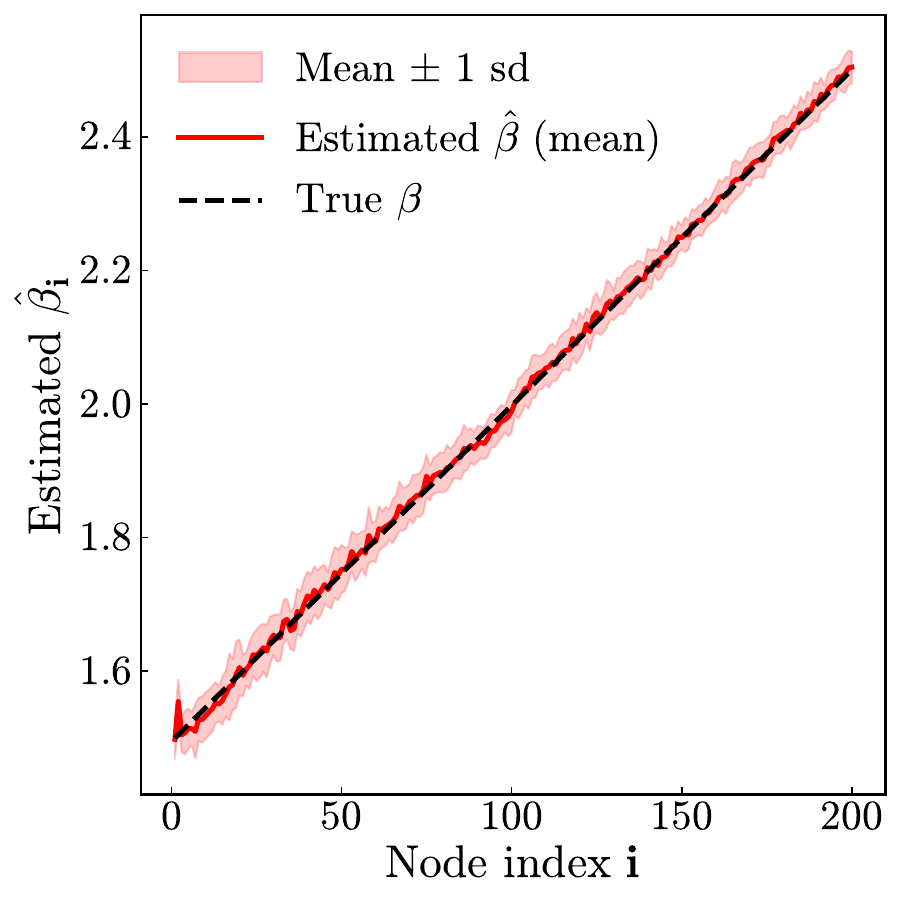}
    \caption{$\beta$ estimation}
    \label{fig:left}
\end{subfigure}\hfill
\begin{subfigure}[t]{0.63\textwidth}
    \centering
    \includegraphics[height=5cm]{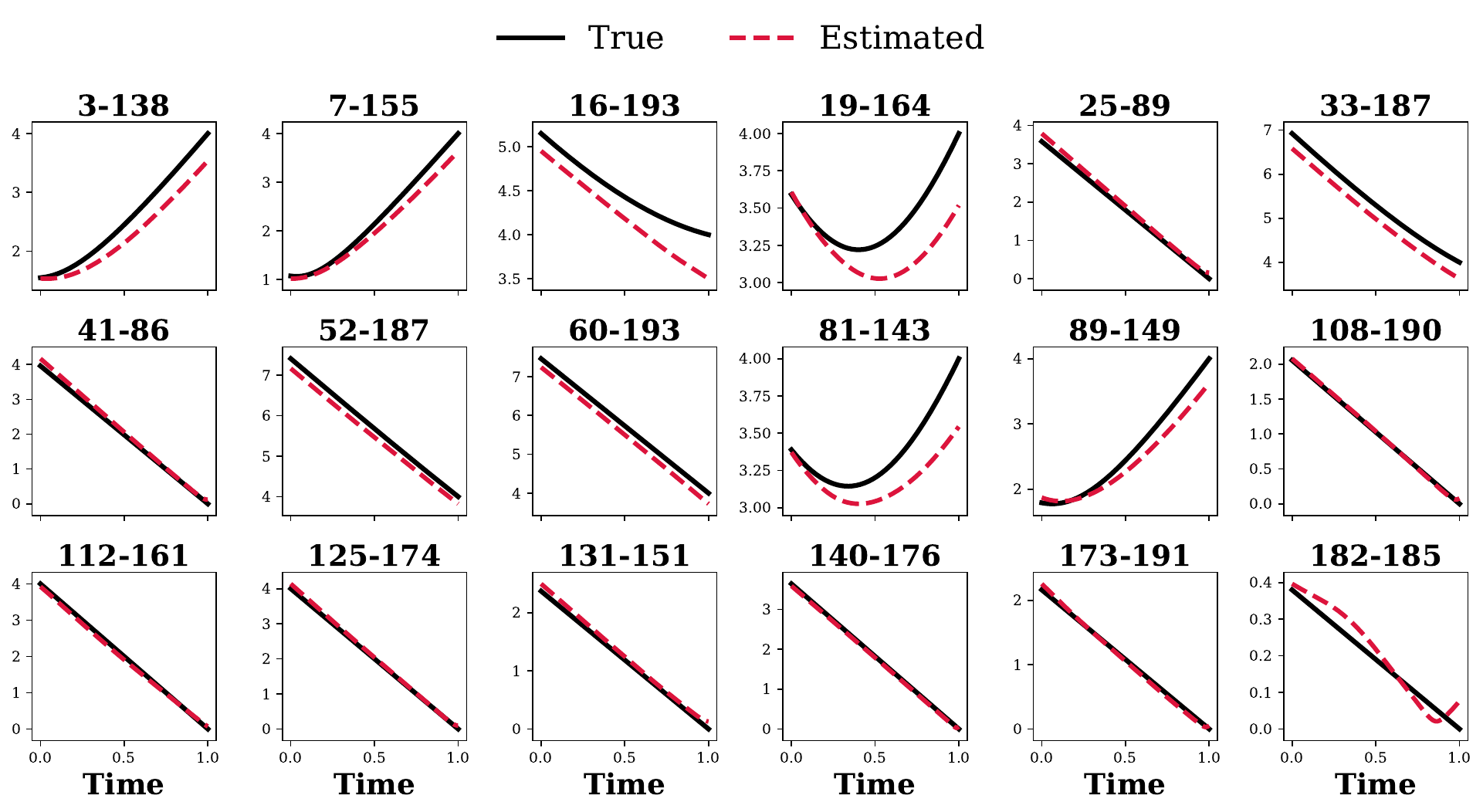}
    \caption{Randomly selected pairwise distances}
    \label{fig:right}
\end{subfigure}
\caption{Estimation of $\beta$ across 30 trials and of the latent distance in one trial under S2. Both components are recovered accurately, with tight uncertainty across all nodes.}
\label{fig:S2}
\end{figure}

In this synthetic experiment, the ground-truth trajectories are straight lines, so stronger smoothing is well aligned with the data-generating process. 
To illustrate the effect of $\rho_2$, Figure~\ref{fig:star_lambda2} displays estimated trajectories for $\rho_2\in\{0.001,1,100\}$ with $\rho_1$ fixed at $0.1$: all three recover the true shapes, while larger $\rho_2$ produces trajectories that are progressively straighter. 

\begin{figure}[ht]
    \centering
    \begin{subfigure}[t]{.319\textwidth}
        \centering
        \includegraphics[width=\linewidth]{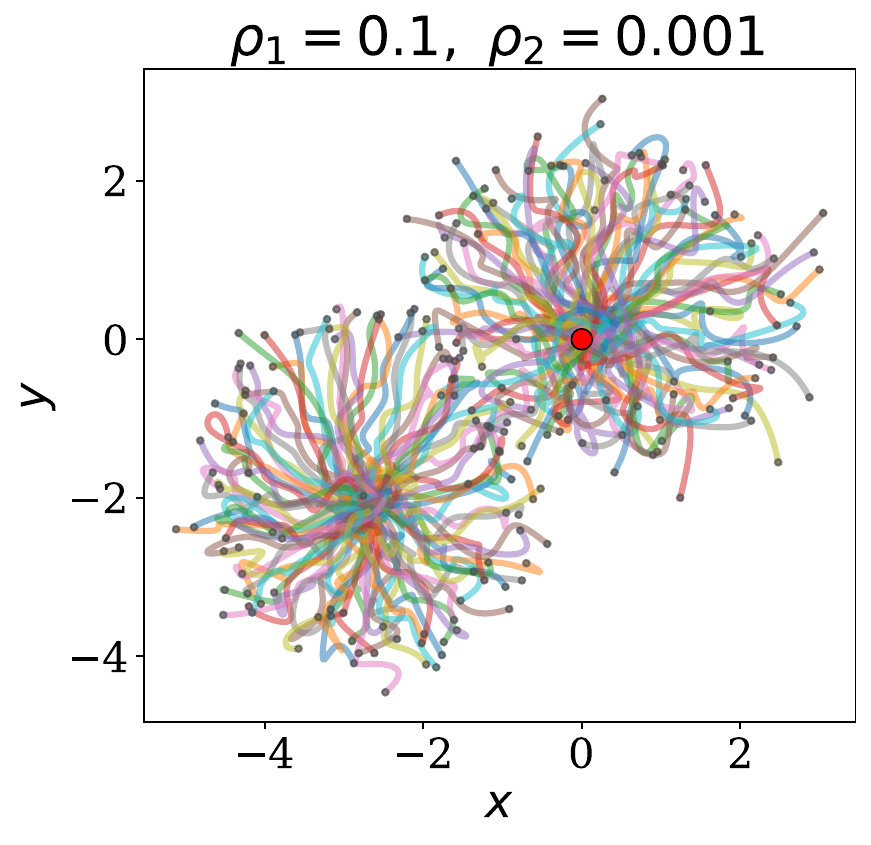}
        \label{fig:lambda2_0p001}
    \end{subfigure}\hfill
    \begin{subfigure}[t]{.34\textwidth}
        \centering
        \includegraphics[width=\linewidth]{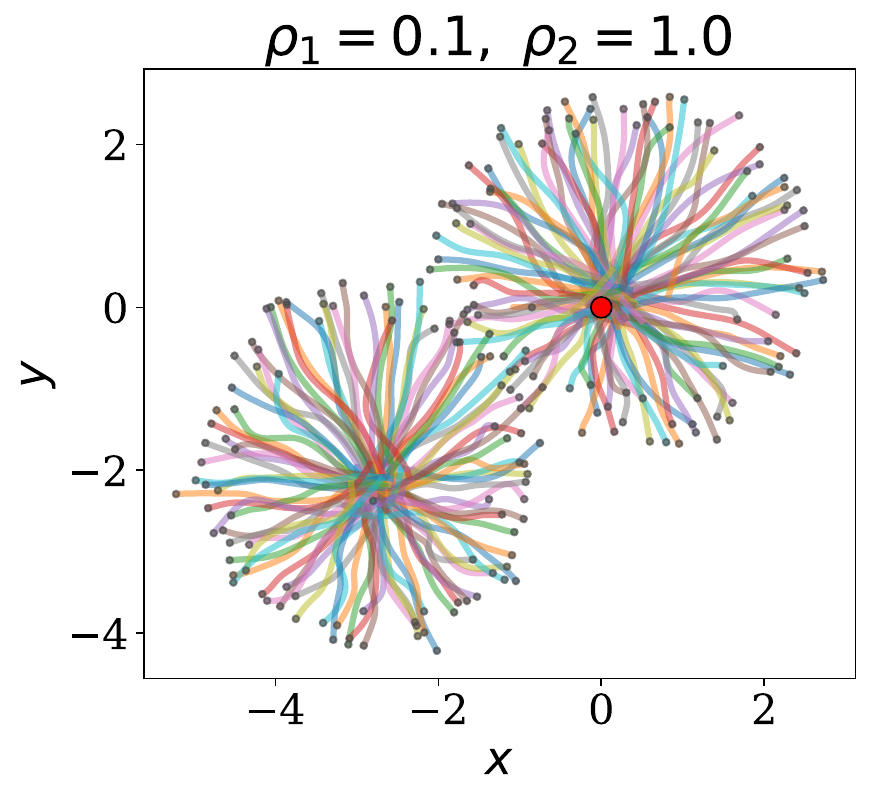}
        \label{fig:lambda2_0p01}
    \end{subfigure}\hfill
    \begin{subfigure}[t]{.319\textwidth}
        \centering
        \includegraphics[width=\linewidth]{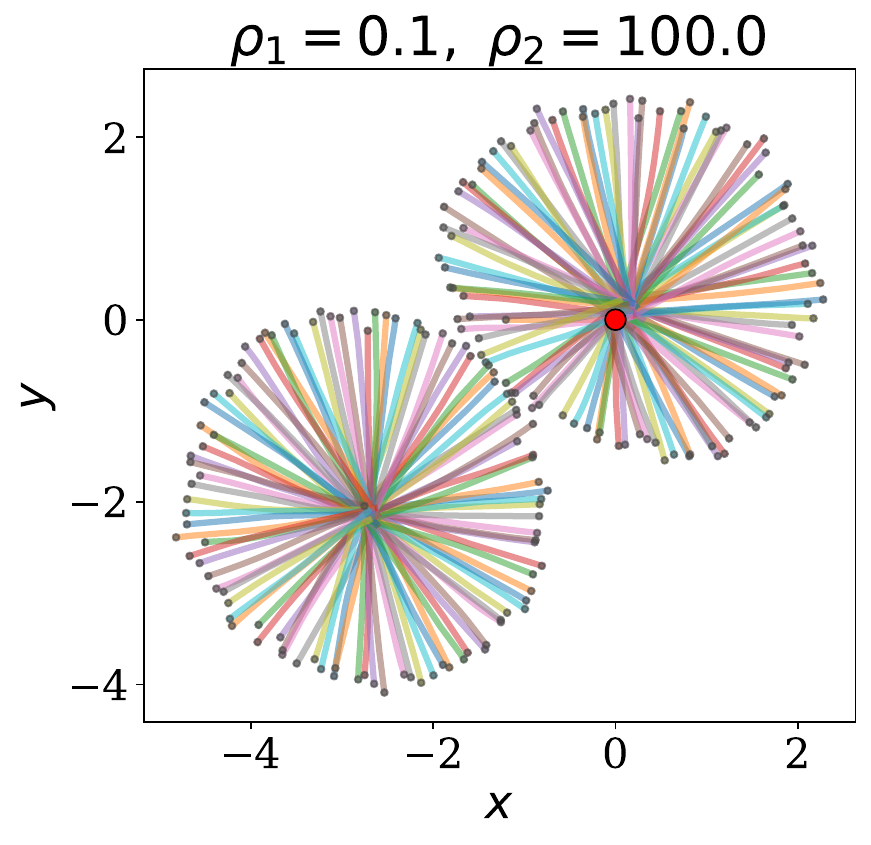}
        \label{fig:lambda2_0p1}
    \end{subfigure}
    \vspace{-2em}
\caption{Estimated trajectories in S2 for increasing $\rho_2$ ($\rho_1 = 0.1$). Larger $\rho_2$ yields progressively straighter and smoother trajectories while preserving the true shapes.}
    \label{fig:star_lambda2}
\end{figure}

Similar to S1, Figure~\ref{fig:S2} shows the estimation of $\beta$ over 30 replications and the true and estimated distance of 18 randomly selected pairs. Table~\ref{tab:S2} reports the mean and standard error over 30 trials, DLS-PP achieves the good predictive performance across most metrics.

\begin{table}[t]
\centering
\caption{(S2) Comparison of methods on in-sample fit and two out-of-sample scenarios.}
\label{tab:S2}
\footnotesize
\renewcommand{\arraystretch}{1.15}
\begin{tabular}{llccccc}
\toprule
Scenario & Metric & CLPM & DIPP & \makecell[c]{Poisson-$\beta$\\model} & \makecell[c]{Homogeneous\\baseline} & DLS-PP \\
\midrule
\multirow{3}{*}{In-sample}
  & IMSE            & \makecell{0.7472 \\ (0.2110)} & N/A & N/A  & \makecell{\textbf{0.2442} \\ (0.2012)} & \makecell{0.5966 \\ (0.5766)} \\
\cmidrule(l){2-7}
  & KS              & \makecell{0.3862 \\ (0.0159)} & \makecell{0.0269 \\ (0.0004)} & \makecell{0.3278 \\ (0.0005)} & \makecell{0.0334 \\ (0.0030)} & \makecell{\textbf{0.0098} \\ (0.0081)} \\
  \cmidrule(l){2-7}
  & Err $\beta$     & \makecell{1.1504 \\ (0.0495)} & N/A & \makecell{0.8697 \\ (0.0012)} & \makecell{2.2769 \\ (0.0132)} & \makecell{\textbf{0.0532} \\ (0.0256)} \\
\midrule
\multirow{4}{*}{\makecell[l]{Out-of-sample\\Persistence}}
  & IMSE            & \makecell{1.1213 \\ (0.3785)} & N/A & N/A & \makecell{0.1816 \\ (0.0779)} & \makecell{\textbf{0.0851} \\ (0.0168)} \\
\cmidrule(l){2-7}
  & KS              & \makecell{0.4881 \\ (0.0089)} & \makecell{0.0999 \\ (0.0024)} & \makecell{0.4382 \\ (0.0009)} & \makecell{0.0609 \\ (0.0121)} & \makecell{\textbf{0.0381} \\ (0.0134)} \\
\cmidrule(l){2-7}
  & \makecell[l]{Log-\\likelihood} & \makecell{$2.821 \times 10^{5}$ \\ ($4.452 \times 10^{3}$)} & \makecell{$3.205 \times 10^{5}$ \\ ($2.47 \times 10^{3}$)} & \makecell{$2.812 \times 10^{5}$ \\ (789)} & \makecell{$3.854 \times 10^{5}$ \\ ($1.27 \times 10^{3}$)} & \makecell{$\mathbf{3.932 \times 10^{5}}$ \\ ($1.5 \times 10^{3}$)} \\
\cmidrule(l){2-7}
  & Count           & \makecell{4.3695 \\ (0.0311)} & \makecell{4.3120 \\ (0.0559)} & \makecell{4.0769 \\ (0.0142)} & \makecell{2.4631 \\ (0.0400)} & \makecell{\textbf{1.8305} \\ (0.0394)} \\
\midrule
\multirow{4}{*}{\makecell[l]{Out-of-sample\\Constant-velocity}}
  & IMSE            & \makecell{4.3763 \\ (1.7896)} & N/A & N/A & \makecell{0.1657 \\ (0.1529)} & \makecell{\textbf{0.0930} \\ (0.0539)} \\
\cmidrule(l){2-7}
  & KS              & \makecell{0.5478 \\ (0.0090)} & \makecell{0.1884 \\ (0.0030)} & N/A & \makecell{0.0523 \\ (0.0173)} & \makecell{\textbf{0.0276} \\ (0.0184)} \\
\cmidrule(l){2-7}
  & \makecell[l]{Log-\\likelihood} & \makecell{$1.937 \times 10^{5}$ \\ ($1.030 \times 10^{4}$)} & \makecell{$-3.567 \times 10^{13}$ \\ ($5.04 \times 10^{12}$)} & N/A & \makecell{$3.844 \times 10^{5}$ \\ ($1.53 \times 10^{3}$)} & \makecell{$\mathbf{3.936 \times 10^{5}}$ \\ ($1.62 \times 10^{3}$)} \\
\cmidrule(l){2-7}
  & Count           & \makecell{4.5850 \\ (0.0305)} & \makecell{$1.614 \times 10^{10}$ \\ ($1.80 \times 10^{9}$)} & N/A & \makecell{2.5129 \\ (0.0454)} & \makecell{\textbf{1.8067} \\ (0.0534)} \\
\bottomrule
\end{tabular}
\end{table}


%% file: tex_files/6_application.tex
\providecommand{\needs}[1]{\medskip\noindent\textbf{[NEEDS RESULT --- #1]}\medskip}
\providecommand{\flagg}[1]{\medskip\noindent\textbf{[FLAG --- #1]}\medskip}

\section{Application}
\label{sec:app}
 
\subsection{Data}
\label{sec:app-data}

We study cooperative diplomatic interaction among 60 major economies over 1995--2022. Using GDP adjusted for purchasing power parity (GDP PPP) from the World Bank World Development Indicators, we therefore rank all country economies in each year, and retain the $60$ with the smallest median rank over 1995--2022. 
Samples of the largest $N$ economies constructed in this way are standard in empirical macroeconomics and international political economy \citep{diGiovanniLevchenko2012, diGiovanniLevchenko2013, Ossa2015}. More details and the selected list are given in Supplementary Section 6.
 
Interactions are measured with ICEWS coded event data \citep{Boschee2015ICEWS}, a machine-coded archive of more than 17 million events from 1995 onward, extracted from international news sources by natural language processing and coded to the Conflict and Mediation Event Observations (CAMEO) ontology.
Each record carries a date, a source actor, a target actor, an event type and a Goldstein score measuring cooperative or conflictual intensity. Our model uses only the actor pair and the timestamp. ICEWS has been used extensively for crisis forecasting \citep{OBrien2010ICEWS}, for latent-factor and tensor models of country--dyad event counts \citep{ScheinPaisleyBleiWallach2015, ScheinZhouBleiWallach2016} and for constructing measures of state behaviour \citep{BagozziBerlinerWelch2021, KorkmazEtAl2016CivilUnrest}. 
Since coverage derives from news reporting it is uneven, with events in less covered regions underreported \citep{BagozziEtAl2019Underreporting}, so event volume reflects media attention as well as political activity and our estimates describe the reported interaction network. 
 
We retain events whose top-level CAMEO category lies between 01 and 08, comprising verbal and material cooperation, and exclude the conflictual categories 09--20, so that latent proximity admits a single unambiguous reading as cooperative engagement. For each unordered pair of countries we take the recorded event dates as the realization $\mathcal{H}_{ij}$ of the dyadic point process, pooling the two directions, and rescale calendar time to the unit interval. This yields $n = 60$ nodes and a total of $m$ retained events.
 

Treating dyads as unordered is a modelling choice, and the aggregate data support it. Writing $W_{ij}$ for the total number of events from sender $i$ to receiver $j$, as displayed in Figure~\ref{fig:icews_heatmap_total}, the weighted reciprocity $r_w \;=\; 1-\frac{\sum_{i<j}\lvert W_{ij}-W_{ji}\rvert}{\sum_{i\neq j} W_{ij}} \;=\; 0.934$ leaves approximately $6.6\%$ of total volume as net directional imbalance, and the correlation between $\log(1+W_{ij})$ and $\log(1+W_{ji})$ across dyads is $0.987$. Figure~\ref{fig:icews_heatmap_total} shows that recorded events concentrate among high-GDP countries, consistent with heterogeneous media coverage in automated event data. This is the degree heterogeneity that the activity parameters $\beta_i$ in \eqref{eq:intensity} are designed to absorb. 
 
\subsection{Baseline comparison}
\label{sec:app-spec}
 
We fit the DLS-PP model \eqref{eq:intensity} and four competing methods as described in Section~\ref{sec:simulation}. Details of tuning are reported in Supplementary Section 5. Table~\ref{tab:application-comparison} compares model performance under in-sample and out-of-sample evaluation. For the out-of-sample scenario, we hold out the final month across all country pairs (93,773 events) while fitting the model on the remaining 4,198,928 events. Our model achieves the best performance across all metrics, excelling in both in-sample fit and predictive capability.
 
\begin{table}[ht]
\centering
\caption{Comparison of methods on application ICEWS dataset across in-sample fit and out-of-sample scenarios.}
\label{tab:application-comparison}
\footnotesize
\renewcommand{\arraystretch}{1.15}
\begin{tabular}{llccccc}
\toprule
Scenario & Metric & CLPM & DIPP & \makecell[c]{Poisson-$\beta$\\model} & \makecell[c]{Homogeneous\\baseline} & DLS-PP \\
\midrule
In-sample 
  & KS & 0.7278 & 0.6868 & 0.3121 & 0.4360 & \textbf{0.2435} \\
\midrule
\multirow{3}{*}{\makecell[l]{Out-of-sample\\Persistence}}
  & KS & 0.8878 & 0.7055 & 0.3622 & 0.4540 & \textbf{0.2777} \\
\cmidrule(l){2-7}
  & \makecell[l]{Log-\\likelihood} & $3.16 \times 10^5$ & $7.03 \times 10^5$ & $7.00 \times 10^5$ & $6.52 \times 10^5$ & {\boldmath $7.25 \times 10^5$} \\
\cmidrule(l){2-7}
  & Count & 124.65 & 188.93 & 87.80 & 114.36 & \textbf{79.43} \\
\midrule
\multirow{3}{*}{\makecell[l]{Out-of-sample\\Constant-velocity}}
  & KS & 0.8916 & 0.7057 & N/A & 0.4576 & \textbf{0.2945} \\
\cmidrule(l){2-7}
  & \makecell[l]{Log-\\likelihood} & $3.04 \times 10^5$ & $-8.42 \times 10^7$ & N/A & $6.40 \times 10^5$ & {\boldmath $7.15 \times 10^5$} \\
\cmidrule(l){2-7}
  & Count & 124.65 & $8.1892 \times 10^5$ & N/A & 115.09 & \textbf{84.90} \\
\bottomrule
\end{tabular}
\end{table}

\subsection{Estimation in DLS-PP}
\label{sec:app-geometry}
We now examine the in-sample estimates under DLS-PP to interpret the underlying network dynamics.
Figure~\ref{fig:node_beta} displays the estimated activity parameters $\hat{\beta}_i$, which measure each country's baseline propensity for diplomatic interaction. Since countries are ordered by descending GDP rank, the general downward slope confirms that diplomatic volume broadly scales with economic capacity, led by major powers like the United States, China, and the Russian Federation. Conversely, established middle powers such as Canada, the Netherlands, and Sweden exhibit $\hat{\beta}_i$ values noticeably below their economic rank, capturing a more low-key diplomatic posture relative to their output.

\begin{figure}[htbp]
    \centering
    \includegraphics[width=0.95\linewidth]{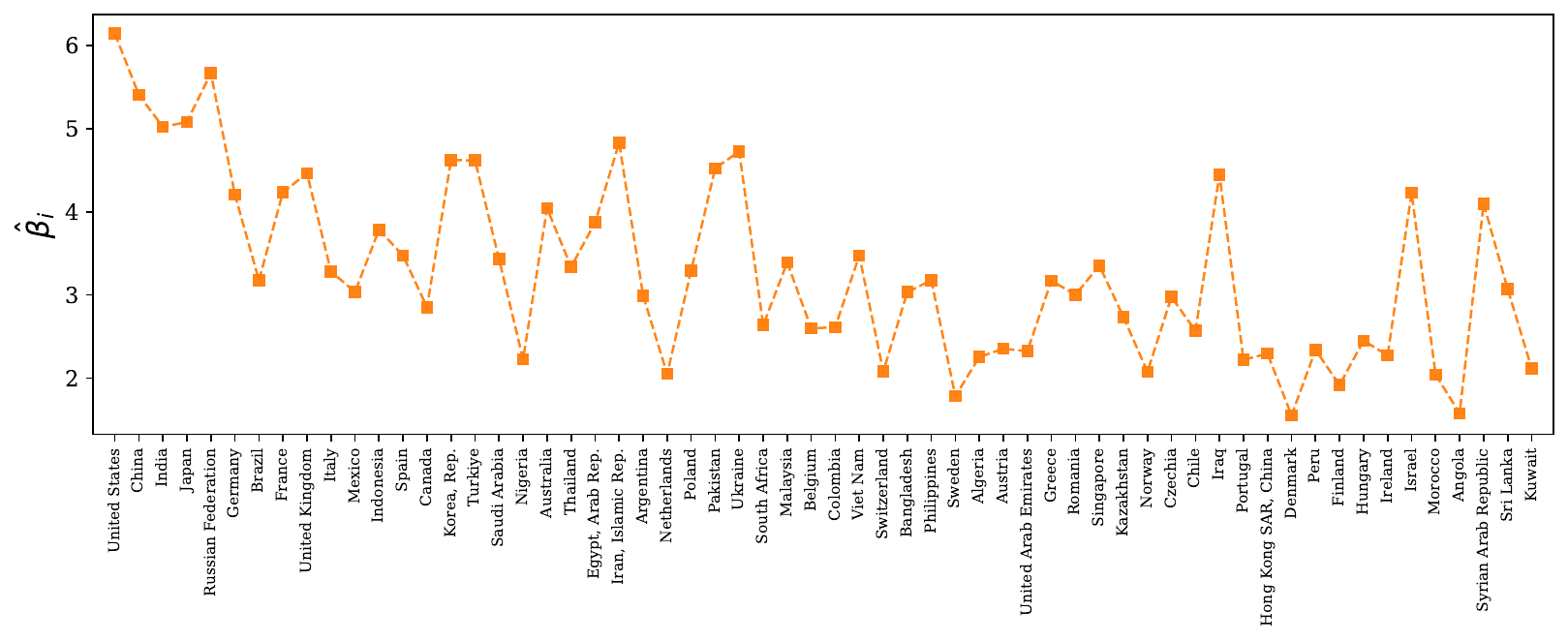}
    \caption{Estimated activity parameters $\hat\beta_i$, with countries ordered by GDP rank.}
    \label{fig:node_beta}
\end{figure}

To illustrate the continuous latent space geometry, Figure~\ref{fig:snap} displays four temporal snapshots with the United States fixed as the origin anchor. Notably, China, Japan, and South Korea display tight clustering and parallel migration toward the core, illustrating how regional interdependence drives coordinated diplomatic trajectories in latent space. Other key actors, such as Russia and Germany, trace unique paths reflecting broader realignment in international relations. The animation is provided with the submission.

\begin{figure}[ht]
    \centering
    \includegraphics[width=1\linewidth]{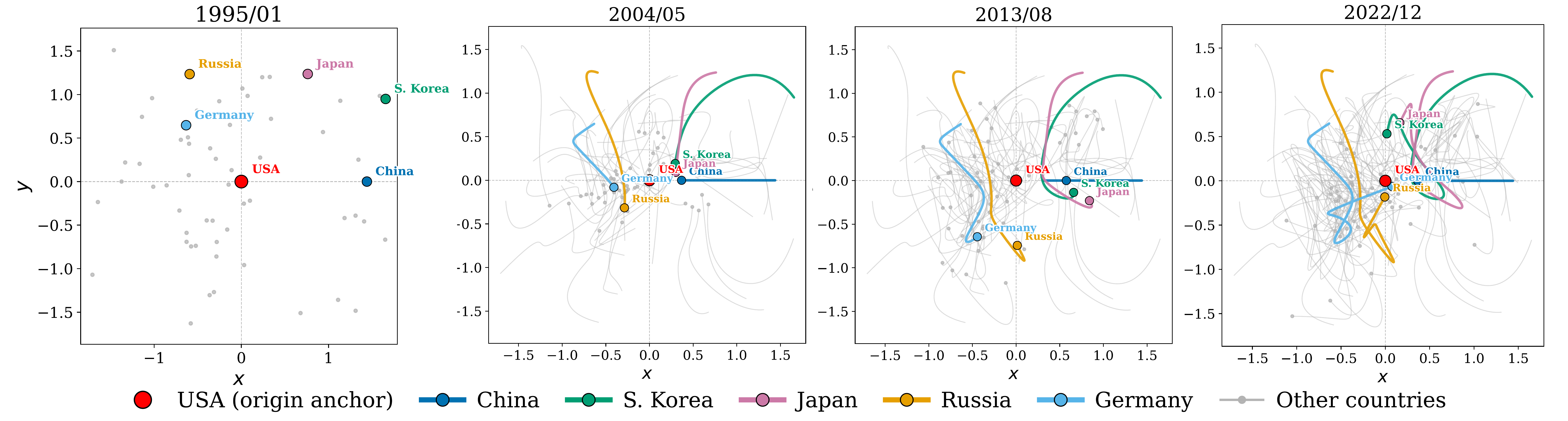}
\caption{Estimated continuous-time latent trajectories at four equally spaced dates, with the United States anchored at the origin. Highlighted paths showcase dynamic spatial alignments among six key nations over time.}
    \label{fig:snap}
\end{figure}

We report the length and arc of each estimated trajectory in Table~\ref{tab:country_length_summary}. Length measures how far a country moved in the latent space between 1995 and 2022, and arc measures how sharply it changed direction. Countries high on both, such as Ukraine, Iraq and Poland, shifted their international relationships more abruptly. Bangladesh moved just as far but far more smoothly, changing gradually rather than abruptly. Canada, the United Kingdom, the Netherlands and France barely moved at all, staying stable throughout.

\begin{table}[htbp]
\centering
\caption{Top and bottom 10 countries by trajectory length and arc.
Length is the path length $\int_0^T \|\dot{\hat z}_i(t)\|\,dt$. Arc is the
curvature penalty $\int_0^T \|\ddot{\hat z}_i(t)\|^2\,dt$, reported in units
of $10^3$.}
\label{tab:country_length_summary}
\setlength{\tabcolsep}{3pt}
\resizebox{\textwidth}{!}{%
\begin{tabular}{lrlrlrlr}
\toprule
\multicolumn{2}{c}{\textbf{Top 10 Length (Mobile)}} & \multicolumn{2}{c}{\textbf{Bottom 10 Length (Static)}} & \multicolumn{2}{c}{\textbf{Top 10 Arc (Curved)}} & \multicolumn{2}{c}{\textbf{Bottom 10 Arc (Straight)}} \\
\cmidrule(lr){1-2} \cmidrule(lr){3-4} \cmidrule(lr){5-6} \cmidrule(lr){7-8}
Country & Length & Country & Length & Country & Arc & Country & Arc \\
\midrule
Ukraine              & 7.01 & United States$^{\dagger}$  & 0.00 & Poland               & 9.98 & United States$^{\dagger}$ & 0.00 \\
Iraq                 & 5.57 & Canada                     & 1.46 & Indonesia            & 8.61 & Denmark                   & 0.72 \\
Poland               & 4.92 & United Kingdom             & 1.62 & Iraq                 & 8.35 & Netherlands               & 0.83 \\
Pakistan             & 4.75 & Netherlands                & 1.67 & Brazil               & 8.11 & Austria                   & 0.90 \\
Sri Lanka            & 4.59 & Ireland                    & 1.73 & Philippines          & 7.66 & Ireland                   & 0.97 \\
Brazil               & 4.28 & China$^{\ddagger}$         & 1.78 & Turkiye              & 7.56 & Israel                    & 1.28 \\
Syrian Arab Republic & 4.17 & Switzerland                & 1.92 & Argentina            & 6.93 & Canada                    & 1.37 \\
Iran, Islamic Rep.   & 4.09 & Italy                      & 1.95 & Romania              & 6.87 & Sweden                    & 1.40 \\
Saudi Arabia         & 4.01 & France                     & 2.03 & Syrian Arab Republic & 6.81 & Italy                     & 1.42 \\
Korea, Rep.          & 4.01 & Israel                     & 2.08 & Ukraine              & 6.55 & Bangladesh                & 1.52 \\
\bottomrule
\end{tabular}%
}
\vspace{0.5em}
\begin{minipage}{\textwidth}
\footnotesize
$^{\dagger}$The United States is the origin anchor, $z_1(t)\equiv(0,0)$, so its
length and arc are identically zero by construction rather than by estimation.
$^{\ddagger}$China is the axis anchor, $y_2(t)\equiv 0$, so it moves along a
line and its values are not directly comparable with the remaining countries.
\end{minipage}
\end{table}

%% file: tex_files/7_discussion.tex
\section{Conclusion}
\label{sec:discussion}

We have developed a dynamic latent space framework for continuous-time event networks (DLS-PP), supported by an scalable estimation pipeline and an effective-degrees-of-freedom selection criterion. Methodologically, a key feature of the model is the separation of node-specific activity levels ($\beta_i$) from dynamic spatial geometry ($z_i(t)$). By preventing baseline interaction rates from masquerading as central spatial positioning, this construction eliminates degree-driven core-periphery artifacts, ensuring that estimated trajectories reflect genuine shifts in dyadic alignment, which is a crucial distinction for isolating mobile geopolitical actors from stationary institutional anchors.

More broadly, our approach serves as a modular template for continuous-time spatial modeling. Because the spline trajectories, dyadic minibatching, warm starts, and geometric penalties interact solely through the log-intensity, the computational architecture extends directly to discrete network snapshots, weighted edge counts, or self-exciting Hawkes processes. Future extensions include allowing activity parameters to evolve continuously over time, integrating gravity model covariates directly into the intensity function, and leveraging Laplace curvature approximations for formal trajectory uncertainty quantification.

%% file: tex_files/8_AIdeclare.tex
\begingroup \small \linespread{0.95}\selectfont 

\paragraph*{AI Usage Disclosure.} 
The authors used Google Gemini and Anthropic Claude to assist with language refinement, code refinement, plotting implementation, and literature-search support during manuscript preparation. These tools were used only for non-substantive implementation and presentation tasks. In particular, Claude was used to generate the illustrative trajectories shown in Figure 2. All AI-assisted text, code, and outputs were reviewed and revised by the authors, who developed the manuscript's structure, substantive arguments, and analysis and take full responsibility for its accuracy and integrity.
\par \endgroup